\documentclass[12pt]{article}
\usepackage[bookmarks=true, colorlinks=true, linkcolor=blue, citecolor=blue, breaklinks=true]{hyperref}
\usepackage{amsfonts}
\usepackage{rotating}
\usepackage{booktabs}
\usepackage{graphicx, graphics}
\usepackage{caption}
\usepackage{amsmath, amsthm, amssymb, natbib}
\usepackage{multirow}
\usepackage{tabularx}
\usepackage{ellipsis}

\usepackage[usenames]{color}
\usepackage[normalem]{ulem}
\definecolor{mygreen}{rgb}{0,.5,0}

\newtheorem{theorem}{Theorem}[section]

\newtheorem{lemma}{Lemma}[section]

\newtheorem{definition}{Definition}[section]

\newtheorem{example}{Example}[section]

\DeclareMathOperator*{\argmin}{arg\,min}

\newcommand{\VaR}{{\text{VaR}}}
\newcommand{\ES}{{\text{ES}}}
\newcommand{\MS}{{\text{MS}}}

\newcommand{\R}{{\mathbb{R}}}
\newcommand{\LInf}{\mathcal{L}^{\infty}(\Omega, \mathcal{F}, P)}
\newcommand{\PSet}{\mathcal{P}}

\newcommand{\X}{\mathcal{X}}
\newcommand{\be}{\begin{equation}}
\newcommand{\ee}{\end{equation}}
\newcommand{\bee}{\begin{equation*}}
\newcommand{\eee}{\end{equation*}}

\begin{document}

\title{On the Measurement of Economic Tail Risk\thanks{We are grateful to the seminar and conference participants at Central University of Finance and Economics, National University of Singapore, Peking University, University of Waterloo, University Paris VII, The International Congress on Industrial and Applied Mathematics 2015, 
Econometric Society China Meeting 2014, Econometric Society Asian Meeting 2014, Econometric Society European Meeting 2014,
INFORMS Annual Meeting 2013, INFORMS Annual Meeting 2014, Quantitative Methods in Finance Conference 2013, RiskMinds Asia Conference 2013, Risk and Regulation Workshop 2014, the 41st European Group of Risk and Insurance Economists Seminar, and the 8th World Congress of the Bachelier Finance Society for their helpful comments and discussion. This research is supported by the University Grants Committee of HKSAR of China, and the Department of Mathematics of HKUST. This research was partially completed while Xianhua Peng was visiting the Institute for Mathematical Sciences and the Center for Quantitative Finance, National University of Singapore, in 2013.}}

\author{Steven Kou\thanks{Risk Management Institute and Department of Mathematics,
National University of Singapore,
21 Heng Mui Keng Terrace,
I3 Building 04-03,
Singapore 119613. Email: matsteve@nus.edu.sg.} \and Xianhua Peng\thanks{Corresponding author.  Department of Mathematics, Hong Kong University of Science and Technology, Clear Water Bay, Kowloon, Hong Kong. Email: maxhpeng@ust.hk. Tel: (852)23587431. Fax: (852)23581643.}}
\date{}

\maketitle

\begin{abstract}
This paper attempts to provide a decision-theoretic foundation for the measurement of economic tail risk, which is not only closely related to utility theory but also relevant to statistical model uncertainty. The main result is that the only risk measures that satisfy a set of economic axioms for the Choquet expected utility and the statistical property of elicitability (i.e. there exists an objective function such that minimizing the expected objective function yields the risk measure) are the mean functional and the median shortfall, which is the median of tail loss distribution. Elicitability is important for backtesting. We also extend the result to address model uncertainty by incorporating multiple scenarios. As an application, we argue that median shortfall is a better alternative than expected shortfall for setting capital requirements in Basel Accords.\\





\end{abstract}

\baselineskip 18pt


\section{Introduction}\label{sec:intro}

 This paper attempts to provide a decision-theoretic foundation for the measurement of economic tail risk. Two important applications are setting insurance premiums and capital requirements for financial institutions.
For example, a widely used class of risk measures for setting insurance risk premiums is proposed by
\citet*{Wang97} based on a set of axioms.
In terms of capital requirements, \citet{Gordy-2003} provides a theoretical foundation for the Basel Accord banking book risk measure, by demonstrating
that under certain conditions the risk measure is asymptotically equivalent to the 99.9\% Value-at-Risk (VaR). VaR is a widely used approach for the measurement of tail risk; see, e.g., \citet*{Duffie-Pan-1997,Duffie-Pan-2001} and \citet*{Jorion07}.

In this paper we focus on two aspects of risk measurement. First,
risk measurement is closely related to utility theories of risk preferences.
The papers that are most relevant to the present paper are
\citet{Schmeidler-1986, Schmeidler-1989},
which extend the expected utility theory by relaxing the independence axiom to the comonotonic independence axiom;
this class of risk preference can successfully explain various violations of the expectated utility theory, such as the Ellsberg paradox.
Second, a major difficulty in measuring tail risk is that
the tail part of a loss distribution is difficult to estimate and hence bears substantial model uncertainty. As emphasized by \citet[][]{Hansen-2013}, ``uncertainty can come from limited data, unknown models and misspecification of those models."

In face of statistical uncertainty, different procedures may be used to forecast the risk measure. It is hence desirable to be able to evaluate which procedure gives a better forecast. The elicitability of a risk measure is a property based on a decision-theoretic framework for evaluating the performance of different forecasting procedures (\citet{Gneiting-2011}). The elicitability of a risk measure means that the risk measure can be obtained by minimizing the expectation of a forecasting objective function (i.e., a scoring rule, see \citet{Jose-Winkler-2011}); then, the forecasting objective function can be used for evaluating different forecasting procedures.

Elicitability is closely related to backtesting, whose objective is  to evaluate the performance of a risk forecasting model. If a risk measure is elicitable, then the sample average forecasting error based on the objective function can be used for backtesting the risk measure.
\citet{Gneiting-2011} shows that VaR is elicitable but expected shortfall is not, which ``may challenge the use of the expected shortfall as a predictive measure of risk, and may provide a partial explanation for the lack of literature on the evaluation of expected shortfall
forecasts, as opposed to quantile or VaR forecasts."
\citet*{Linton-2011} propose a backtest for evaluating VaR estimates that delivers more power in finite samples than existing methods
and develop a mechanism to find out why and when a model is misspecified; see also
\citet[][Ch. 6]{Jorion07}.
 \citet*{Linton-Xiao-2013} point out that VaR has an advantage over expected shortfall as the asymptotic inference procedures for VaR ``has the same asymptotic behavior regardless of the thickness of the tails."

The elicitability of a risk measure is also related to the concept of ``consistency" of a risk measure proposed by \citet{Davis-2013}, who shows that VaR exhibits some inherent superiority over other risk measures.

The main result of the paper is that the only risk measures that satisfy both a set of economic axioms proposed by \citet{Schmeidler-1989} and the statistical requirement of elicitability (\citet{Gneiting-2011}) are the mean functional and the median shortfall, which is the median of the tail loss distribution and is also the VaR at a higher confidence level.

A risk measure is said to be robust if
(i) it can accommodate model misspecification (possibly by incorporating multiple scenarios and models)
and (ii) it has statistical robustness, which means that a small deviation in the model or small changes in the data only results in a small change in the risk measurement.
The first part of the meaning of robustness is related to ambiguity and model uncertainty in decision theory.
To address these issues, multiple priors or multiple models may be used; see \citet*{Gilboa-Schmeidler-1989}, \citet*{Maccheroni-06}, and \citet*{Hansen-Sargent-2001, HS-2007}, among others.
We also incorporate multiple models in this paper; see Section \ref{sec:mult_scen}. We add to the this literature by studying the link between risk measures and statistical uncertainty via elicitability. As for the second part of the meaning of robustness,
\citet*{CDS-2010} show that expected shortfall leads to a less robust risk measurement procedure than historical VaR; \citet*{KPH-2006, KPH-2013}
propose a set of axioms for robust external risk measures, which include VaR.

There has been a growing literature on capital requirements for banking regulation and robust risk measurement. \citet{GK-2013} investigate the design of risk weights to align regulatory and private objectives in a mean-variance framework for portfolio selection. \citet*{GX-2014} develop a framework for quantifying the impact of model error and for measuring and minimizing risk in a way that is robust to model error. \citet*{Keppo20102192} show that the Basel II market risk requirements may have the unintended consequence of postponing banks' recapitalization and hence increasing banks' default probability. We add to this literature by applying our theoretical results to the study on which risk measure may be more suitable for setting capital requirements in Basel Accords; see Section \ref{sec:app}.

Important contribution to measurement of risk based on economic axioms includes \citet{Aumann-Serrano-2008}, \citet{Foster-Hart-2009, Foster-Hart-2013}, and \citet{Hart-2011}, which
study risk measurement of gambles (i.e., random variables with positive mean and taking negative values with positive probability). This paper complements their results by linking economic axioms for risk measurement with statistical model uncertainty; in addition, our approach
focuses on the measurement of tail risk for general random variables. Thus, the risk measure considered in this paper has a different objective.

The remainder of the paper is organized as follows. Section \ref{sec:main_results} presents the main result of the paper. In Section \ref{sec:mult_scen}, we propose to use a scenario aggregation function to combine risk measurements under multiple models. In Section \ref{sec:app}, we apply the results in previous sections to the study of Basel Accord capital requirements. Section \ref{sec:comm} is devoted to relevant comments. 

%
%

%

\section{Main Results}\label{sec:main_results}

\subsection{Axioms and Representation}



Let $(\Omega, \mathcal{F}, P)$ be a probability space that describes the states and the probability of occurrence of states at a future time $T$.
Assume the probability space is large enough so that one can define a random variable uniformly distributed on [0,1].
Let a random variable $X$ defined on the probability space denote the random loss of a portfolio of financial assets that will be realized at time $T$. Then $-X$ is the random profit of the portfolio. Let $\mathcal{X}$ be a set of random variables that include all bounded random variables, i.e., $\mathcal{X}\supset \LInf$, where $\LInf:=\{X\mid \text{there exists}\ M<\infty\ \text{such that}\ |X|\leq M, \text{a.s. P}\}$. A risk measure $\rho$ is a functional defined on $\mathcal{X}$ that maps a random variable $X$ to a real number $\rho(X)$. The specification of $\X$ depends on $\rho$; in particular, $\X$ can include unbounded random variables. For example,
if $\rho$ is variance, then $\X$ can be specified as $\mathcal{L}^2(\Omega, \mathcal{F}, P)$; if $\rho$ is VaR, then $\X$ can be specified as the set of all random variables.

An important relation between two random variables is comonotonicity (\citet{Schmeidler-1986}):
Two random variables $X$ and $Y$ are said to be comonotonic, if $(X(\omega_1)-X(\omega_2))(Y(\omega_1)-Y(\omega_2))\geq 0$, $\forall \omega_1, \omega_2\in\Omega$.
Let $X$ and $Y$ be the loss of two portfolios, respectively. Suppose that there is a representative agent in the economy and he or she prefers the profit $-X$ to the profit $-Y$. If the agent is risk averse, then
his or her preference may imply that $-X$ is less risky than $-Y$. Motivated by this,
we propose the following set of axioms, which are based on the axioms for the Choquet expected utility (\citet{Schmeidler-1989}), for the risk measure $\rho$.


\noindent\textbf{Axiom A1.} Comonotonic independence: for all pairwise comonotonic random variables $X, Y, Z$ and for all $\alpha\in(0, 1)$, $\rho(X)<\rho(Y)$ implies that $\rho(\alpha X+(1-\alpha)Z)<\rho(\alpha Y+(1-\alpha)Z)$.


\noindent\textbf{Axiom A2.} Monotonicity: $\rho (X)\leq \rho (Y)$,
if $X\leq Y$.

\noindent\textbf{Axiom A3.} Standardization: $\rho(x\cdot 1_{\Omega})=sx$, for all $x\in\R$, where $s>0$ is a constant.


\noindent\textbf{Axiom A4.} Law invariance: $\rho (X)=\rho (Y)$ if $X$ and $Y$ have the same distribution.

\noindent\textbf{Axiom A5.} Continuity: $\lim_{M\to\infty}\rho (\min(\max(X, -M), M))=\rho (X)$, $\forall X$.

%
Axiom A1 corresponds to the comonotonic independence axiom for the Choquet expected utility risk preferences (\citet{Schmeidler-1989}). Axiom A2 is a minimum requirement for a reasonable risk measure.
Axiom A3 with $s=1$ is used in \citet{Schmeidler-1986};
the constant $s$ in Axiom A3 can be related to the ``countercyclical indexing" risk measures proposed in \citet{Gordy-2006}, where a time-varying multiplier $s$ that increases during booms and decreases during recessions is used to dampen the procyclicality of capital requirements; see also \citet*{Brun-Peder2009}, \citet*{BCGPS2009}, and \citet*{Adrian-Shin-2014}. Axiom A4 is standard for a law invariant risk measure. Axiom A5 states that the risk measurement of an unbounded random variable can be approximated by that of  bounded random variables.

A function $h:[0, 1]\to [0, 1]$ is called a distortion function if $h(0)=0$, $h(1)=1$, and $h$ is increasing; $h$ need not be left or right continuous. As a direct application of the results in \citet{Schmeidler-1986}, we obtain the following representation of a risk measure that satisfies Axioms A1-A5.

\begin{lemma}\label{lemma:risk_measure_Schmeidler}
Let $\mathcal{X}\supset \LInf$ be a set of random variables ($\mathcal{X}$ may include unbounded random variables). A risk measure $\rho: \mathcal{X}\to \mathbb{R}$ satisfies Axioms A1-A5 if and only if there exists a distortion function $h(\cdot)$ such that
  \begin{align}
  \rho(X)&=s\int X\,d(h\circ P)\label{equ:e_78}\\
  &=s\int_{-\infty}^0 (h(P(X>x))-1)dx+s\int_0^{\infty} h(P(X>x))dx,\ \forall X\in\X,\label{equ:e_60}
  \end{align}
where the integral in \eqref{equ:e_78} is the Choquet integral of $X$ with respect to the distorted non-additive probability $h\circ P(A):=h(P(A))$, $\forall A\in\mathcal{F}$.
\end{lemma}
\begin{proof} See Appendix \ref{app:proof_rm_rep}.
\end{proof}

%


Lemma \ref{lemma:risk_measure_Schmeidler} extends the representation theorem in \citet*{Wang97} as
the requirement of $\lim_{d\rightarrow 0}\rho ((X-d)^{+})=\rho (X^+)$ in their continuity axiom is not needed here.\footnote{
The axioms used in \citet*{Wang97}, including a comonotonic additivity axiom, imply Axioms A1-A5. More precisely,
let $\mathbb{Q}$ and $\mathbb{Q}^+$ denote the set of rational numbers and positive rational numbers, respectively. Without loss of generality, suppose $s=1$ in Axiom A3. (i) Their comonotonic additivity axiom  implies that $\rho(\lambda X)=\lambda\rho(X)$ for any $X$ and $\lambda\in\mathbb{Q}^+$, which in combination with their standardization axiom $\rho(1)=1$ implies $\rho(\lambda)=\lambda\rho(1)=\lambda$, $\lambda\in\mathbb{Q}^+$. Since $\rho(-\lambda)+\rho(\lambda)=\rho(0)=0$, it follows that $\rho(\lambda)=\lambda$, $\forall \lambda\in\mathbb{Q}$. Then for any $\lambda\in\mathbb{R}$, there exists $\{x_n\}\subset \mathbb{Q}$ and $\{y_n\}\subset\mathbb{Q}$ such that $x_n\downarrow \lambda$ and $y_n\uparrow \lambda$. By the monotonic axiom, $x_n=\rho(x_n)\geq \rho(\lambda)\geq \rho(y_n)=y_n$. Letting $n\to\infty$ yields $\rho(\lambda)=\lambda$, $\forall \lambda\in\mathbb{R}$; hence, Axiom A3 holds.
(ii) By the monotonic axiom, $\rho(\min(X, M))\leq \rho(\min(\max(X, -M), M))\leq \rho(\max(X, -M))$. Letting $M\to\infty$ and using the conditions $\rho(\min(X, M))\to\rho(X)$ and $\rho(\max(X, -M))\to\rho(X)$ as $M\to\infty$ in their continuity axiom, \emph{without need of the condition $\lim_{d\rightarrow 0}\rho ((X-d)^{+})=\rho (X^+)$}, Axiom A5 follows.
(iii) We then show positive homogeneity holds, i.e.  $\rho(\lambda X)=\lambda\rho(X)$ for any $X$ and any $\lambda>0$. For any $X$ and $M>0$, denote $X^M:=\min(\max(X, -M), M)$. For any $\epsilon>0$ and $\lambda>0$, there exist $\{\lambda_n\}\subset \mathbb{Q}^+$ such that $\lambda_n\to\lambda$ as $n\to\infty$ and $\lambda_n\rho(X^M)-\epsilon=\rho(\lambda_n X^M-\epsilon)\leq \rho(\lambda X^M)\leq \rho(\lambda_n X^M+\epsilon)=\lambda_n\rho(X^M)+\epsilon$. Letting $n\to\infty$ yields $\lambda\rho(X^M)-\epsilon\leq \rho(\lambda X^M)\leq \lambda\rho(X^M)+\epsilon$, $\forall \epsilon>0$. Letting $\epsilon\downarrow 0$ leads to $\rho(\lambda X^M)=\lambda\rho(X^M)$, $\forall \lambda\geq 0$. Letting $M\to\infty$ and applying Axiom A5 result in $\rho(\lambda X)=\lambda\rho(X)$, $\forall \lambda\geq 0$. Their comonotonic additivity axiom and positive homogeneity imply Axiom A1.
}
Note that in the case of random variables, the corollary in \citet[][]{Schmeidler-1986} requires the random variables to be bounded, but Lemma \ref{lemma:risk_measure_Schmeidler} does not; Axiom A5 is automatically satisfied for bounded random variables.

It is clear from \eqref{equ:e_60} that any risk measure satisfying Axioms A1-A5 is monotonic with respect to first-order stochastic dominance.\footnote{For two random variables $X$ and $Y$, if $X$ first-order stochastically dominates $Y$, then $P(X>x)\geq P(Y>x)$ for all $x$, which implies that for a risk measure $\rho$ represented by \eqref{equ:e_60}, $\rho(X)\geq \rho(Y)$.}
Many commonly used risk measures are special cases of risk measures defined in \eqref{equ:e_60}.

\begin{example}\label{ex:var}
Value-at-Risk (VaR). VaR is a quantile of the loss distribution at
some pre-defined probability level. More precisely,
let $X$ be the random loss with general distribution function $F_X(\cdot)$, which may not be continuous or strictly increasing. For a given $\alpha
\in (0,1]$, VaR of $X$ at level $\alpha$ is defined as 
\begin{equation*}
\VaR_{\alpha}(X):=F_X^{-1}(\alpha)=\inf \{x\mid F_X(x)\geq \alpha \}.
\end{equation*}
For $\alpha=0$, VaR of $X$ at level $\alpha$ is defined to be $\VaR_{0}(X):=\inf\{x\mid F_X(x)>0\}$ and $\VaR_0(X)$ is equal to the essential infimum of $X$. For $\alpha\in(0, 1]$, $\rho$ in \eqref{equ:e_60} is equal to $\VaR_{\alpha}$ if $h(x):=1_{\{x>1-\alpha\}}$; $\rho$ in \eqref{equ:e_60} is equal to $\VaR_0$ if $h(x):=1_{\{x=1\}}$. VaR is monotonic with respect to first-order stochastic dominance.
\end{example}

\begin{example}\label{ex:es}
Expected shortfall (ES). For $\alpha\in[0, 1)$, ES of $X$ at level $\alpha$ is
defined as the mean of the $\alpha$-tail distribution of $X$ (\citet{Tasche02}, \citet*{Rock02}), i.e.,
\begin{equation*}
  \text{ES}_{\alpha}(X):=\text{mean of the}\ \alpha\text{-tail distribution of}\ X
  = \int_{-\infty}^{\infty}x dF_{\alpha,X}(x),\ \alpha\in[0, 1),
\end{equation*}
where $F_{\alpha,X}(x)$ is the $\alpha$-tail distribution defined as (\citet{Rock02}):
\begin{equation*}
  F_{\alpha,X}(x):=\begin{cases}
    0, & \text{for}\ x < \VaR_{\alpha}(X)\\
    \frac{F_X(x)-\alpha}{1-\alpha}, & \text{for}\ x \geq \VaR_{\alpha}(X).
  \end{cases}
\end{equation*}
For $\alpha=1$, ES of $X$ at level $\alpha$ is defined as $\ES_1(X):=F_X^{-1}(1)$.
If the loss distribution $F_X$ is continuous, then $F_{\alpha,X}$ is the same as the conditional distribution of $X$ given that $X\geq \VaR_{\alpha}(X)$; if $F_X$ is not continuous, then $F_{\alpha,X}(x)$ is a slight modification of the conditional loss distribution. For $\alpha\in[0, 1)$, $\rho(X)$ in \eqref{equ:e_60} is equal to $\ES_{\alpha}(X)$ if
$$h(x)=\begin{cases}
            \frac{x}{1-\alpha}, & x\leq 1-\alpha,\\
            1, & x> 1-\alpha.
          \end{cases}$$
For $\alpha=1$, $\rho(X)$ in \eqref{equ:e_60} is equal to $\ES_{1}(X)$ if $h(x)=1_{\{x>0\}}$.
\end{example}

\begin{example}\label{ex:ms}
Median shortfall (MS). As we will see later, expected shortfall has several statistical drawbacks including non-elicitability and non-robustness. To mitigate the problems, one may simply use median shortfall. In contrast to ES which is the mean of the tail loss distribution, MS is the median of the same tail loss distribution. More precisely, MS of $X$ at level $\alpha\in[0, 1)$ is defined as (\citet*{KPH-2013})\footnote{The term ``median shortfall" is also used in \citet{Moscadelli-2004} and \citet{SW-2012} but is respectively defined as $\text{median}[X|X>u]$ for a constant $u$ and $\text{median}[X|X>\VaR_{\alpha}(X)]$, which are different from that defined in \citet*{KPH-2013}. In fact, the definition in the aforementioned second paper is the same as the ``tail conditional median" proposed in \citet*{KPH-2006}.}
\begin{align*}
  \text{MS}_{\alpha}(X):=\text{median of the}\ \alpha\text{-tail distribution of}\ X=F_{\alpha,X}^{-1}(\frac{1}{2})=\inf \{x\mid F_{\alpha, X}(x)\geq \frac{1}{2}\}.
\end{align*}
For $\alpha=1$, $\MS$ at level $\alpha$ is defined as
$\MS_1(X):=F_X^{-1}(1)$. Therefore, MS at level $\alpha$ can \emph{capture the tail risk and considers both the size and likelihood of losses beyond the VaR at level $\alpha$}, because it measures the median of the loss size conditional on that the loss exceeds the VaR at level $\alpha$.
It can be shown that\footnote{Indeed, for $\alpha\in(0, 1)$, by definition,
$\text{MS}_{\alpha}(X)=\inf\{x\mid F_{\alpha, X}(x)\geq \frac{1}{2}\}=\inf\{x\mid \frac{F_{X}(x)-\alpha}{1-\alpha}\geq \frac{1}{2}\}=\inf\{x\mid F_X(x)\geq \frac{1+\alpha}{2}\}=\VaR_{\frac{1+\alpha}{2}}(X)$; for $\alpha=1$, by definition, $\MS_1(X)=F_X^{-1}(1)=\VaR_1(X)$; for $\alpha=0$, by definition, $F_{0, X}=F_X$ and hence $\MS_0(X)=F_X^{-1}(\frac{1}{2})=\VaR_{\frac{1}{2}}(X)$.
}
  \begin{align*}
    &\MS_{\alpha}(X)=\VaR_{\frac{1+\alpha}{2}}(X),\ \forall X,\ \forall \alpha\in[0, 1].
  \end{align*}
Hence, $\rho(X)$ in \eqref{equ:e_60} is equal to $\text{MS}_{\alpha}(X)$ if $h(x):=1_{\{x>(1-\alpha)/2\}}$.

Since $\MS_{\alpha}=\VaR_{(1+\alpha)/2}$, $\MS_{\alpha}$ does not quantify the risk beyond $\VaR_{(1+\alpha)/2}$. However, it is also difficult to know the precise degree to which $\ES_{\alpha}$ quantifies the risk beyond $\VaR_{(1+\alpha)/2}$; in fact, just as $\MS_{\alpha}$, $\ES_{\alpha}$ can also fail to reveal large loss beyond $\VaR_{(1+\alpha)/2}$. For example, fix $c:=\VaR_{\alpha}$ and consider a sequence of $\alpha$-tail distributions $F_{\alpha,n}$ that are mixtures of translated exponential distributions and point mass distributions, which are defined by
\begin{align}\label{equ:counter_example_es}
F_{\alpha,n}(x):=\begin{cases}
  0, & \text{for}\ x < c,\\
  (1-\beta(n))(1-e^{-\lambda(x-c)})+\beta(n)1_{\{n\leq x\}},\ \beta(n):=\frac{\mu}{n-c-\frac{1}{\lambda}}, &\text{for}\ x \geq c,
\end{cases}
\end{align}
where $\lambda, \mu>0$ are constants. In other words, $F_{\alpha,n}$ is the mixture of $c+\exp(\lambda)$ (with probability $(1-\beta(n))$) and the point mass $\delta_n$ (with probability $\beta(n)$). Under $F_{\alpha,n}$, a large loss with size $n$ occurs with a small probability $\beta(n)$. For each $n$, $\ES_{\alpha,n}$, which is the mean of $F_{\alpha,n}$, is always equal to $c+\mu+\frac{1}{\lambda}$; hence, $\ES_{\alpha}$ fails in the same way as $\MS_{\alpha}$ regarding the detection of the large loss with size $n$ which may occur beyond $\VaR_{(1+\alpha)/2}$. This example shows that the degree to which $\ES_{\alpha}$ quantifies the risk beyond $\VaR_{(1+\alpha)/2}$ might also be limited. After all, $\MS_{\alpha}$ and $\ES_{\alpha}$ are respectively the median and the mean of \emph{the same} $\alpha$-tail loss distribution. The information contained in the mean of a distribution might not be more than that contained in the median of the same distribution, and vice versa.
\end{example}

\begin{example}\label{ex:gen_spec_rm}
Generalized spectral risk measures. A generalized spectral risk measure is defined by
\begin{equation}\label{equ:e_80}
\rho_{\Delta}(X):=\int_{(0,1]} F_X^{-1}(u)d\Delta(u),
\end{equation}
where $\Delta$ is a probability measure on $(0, 1]$. The class of risk measures represented by \eqref{equ:e_60} include and are strictly larger than the class of generalized spectral risk measures, as they all satisfy Axioms A1-A5.\footnote{In fact,
for any fixed $u\in(0,1]$, $F_X^{-1}(u)=\VaR_u(X)$ as a functional on $\LInf$ is a special case of the risk measure \eqref{equ:e_60}. By the proof of Lemma \ref{lemma:risk_measure_Schmeidler}, $\VaR_u$
satisfies monotonicity, positive homogeneity, and comonotonic additivity, which implies that $\rho_{\Delta}$ satisfies Axioms A1-A4 for any $\Delta$. On
$\LInf$, $\rho_{\Delta}$ automatically satisfies Axiom A5. On the other hand, for an $\alpha\in(0, 1)$,
the right quantile $q_{\alpha}^+(X):=\inf\{x\mid F_X(x)> \alpha\}$ is a special case of the risk measure defined in \eqref{equ:e_60} with $h(x)$ being defined as $h(x):=1_{\{x\geq 1-\alpha\}}$,
but it can be shown that $q_{\alpha}^+$ cannot be represented by \eqref{equ:e_80}.
Indeed, suppose for the sake of contradiction that
there exists a $\Delta$ such that $q_{\alpha}^+(X)=\rho_{\Delta}(X)$, $\forall X\in\LInf$.
Let $X_0$ have a strictly positive density on its support. Then, $F_{X_0}^{-1}(u)$ is continuous and strictly increases on $(0, 1]$. Let $c>0$ be a constant. Define $X_1=X_0\cdot 1_{\{X_0\leq F_{X_0}^{-1}(\alpha)\}}+(X_0+c)\cdot 1_{\{X_0> F_{X_0}^{-1}(\alpha)\}}$. It follows from $q_{\alpha}^+(X_1)-q_{\alpha}^+(X_0)=\rho_{\Delta}(X_1)-\rho_{\Delta}(X_0)$ that $_{\Delta}((\alpha, 1])=1$, which in combination with the strict monotonicity of $F^{-1}_{X_0}(u)$ implies that $\rho_{\Delta}(X_0)=\int_{(\alpha, 1]}F_{X_0}^{-1}(u)\Delta(du)>F_{X_0}^{-1}(\alpha)=q_{\alpha}^+(X_0)$. This contradicts to $\rho_{\Delta}(X_0)=q^+_{\alpha}(X_0)$.\label{footnote:g_spectral}}
A special case of \eqref{equ:e_80} is the spectral risk measure (\citet{Acerbi-2002}, Definition 3.1), defined as
\begin{equation}\label{equ:spectral_rm}
\rho(X)=\int_{(0,1)} F_X^{-1}(u)\phi(u)du,\ \phi(\cdot)\ \text{is increasing, nonnegative, and}\ \int_{0}^1 \phi(u)du=1.
\end{equation}
Because of the requirement that $\phi$ is increasing, the class of spectral risk measure is much smaller than the class of generalized spectral risk measure defined in \eqref{equ:e_80}.
The distinction between the spectral risk measure and that in \eqref{equ:e_80} is that the former is convex but the latter may not be convex. The convexity requires that the function $\phi$ in \eqref{equ:spectral_rm} is an increasing function. The MINMAXVAR risk measure proposed in \citet{Cherny-Madan-2009} for the measurement of trading performance is a special case of the spectral risk measure, corresponding to a  distortion function
$h(x)=1-(1-x^{\frac{1}{1+\alpha}})^{1+\alpha}$ in \eqref{equ:e_60},
where $\alpha\geq 0$ is a constant.
\end{example}





The class of risk measures satisfying Axioms A1-A5 and the class of law-invariant coherent (convex) risk measures have non-empty intersections but no one is the subset of the other. For example, expected shortfall belongs to both classes; VaR belongs to the former but not the latter. The class of risk measures satisfying Axioms A1-A5 include the class of law-invariant spectral risk measures as a strict subset. For example, VaR belongs to the former but not the latter.
The class of risk measures satisfying Axioms A1-A5 is the same as the class of ``distortion risk measure" proposed in \citet*{Wang97}. The ``distortion risk measures" sometimes refer to the class of risk measures defined in \eqref{equ:e_80}. As we point out in Example \ref{ex:gen_spec_rm}, the class of risk measures defined in \eqref{equ:e_80} is a strict subset of the class of risk measures satisfying Axioms A1-A5.

If a risk measure $\rho$ satisfies Axiom A4 (law invariance), then $\rho(X)$ only depends on $F_X$; hence, $\rho$ induces a statistical functional that maps a distribution $F_X$ to a real number $\rho(X)$. For simplicity of notation, we still denote the induced statistical functional as $\rho$. Namely, we will use $\rho(X)$ and $\rho(F_X)$ interchangeably in the sequel.

\subsection{Elicitability}\label{subsec:elicitability}

The measurement of risk of $X$ using $\rho$ may be viewed as a point forecasting problem, because the risk measurement $\rho(X)$ (or $\rho(F_X)$) summarizes the distribution $F_X$ by a real number $\rho(X)$, just as a point forecast for $X$ does. In practice, the true distribution $F_X$ is unknown and one has to find an estimate $\hat F_X$ for
forecasting the unknown true value $\rho(F_X)$. As one may come up with different procedures to forecast $\rho(F_X)$,
it is an important issue to evaluate which procedure provides a better forecast of $\rho(F_X)$.


The theory of elicitability provides a decision-theoretic foundation for effective evaluation of point forecasting procedures. Suppose one wants to forecast the realization of a random variable $Y$ using a point $x$, without knowing the true distribution $F_Y$. The expected forecasting error is given by
$$ES(x, Y)=\int S(x, y)dF_Y(y),$$
where
$S(x, y): \mathbb{R}^2\to \mathbb{R}$
is a forecasting objective function, e.g., $S(x, y)=(x-y)^2$ or $S(x,y)=|x-y|$. The optimal point forecast corresponding to $S$ is
$$\rho^*(F_Y)=\argmin_x ES(x, Y).$$
For example, when $S(x,y)=(x-y)^2$ and $S(x,y)=|x-y|$, the optimal forecast is the mean functional $\rho^*(F_Y)=E(Y)$ and the median functional $\rho^*(F_Y)=F_Y^{-1}(\frac{1}{2})$, respectively.

A statistical functional $\rho$ is elicitable if there exists a forecasting objective function $S$ such that minimizing the expected forecasting error yields $\rho$. Many statistical functionals are elicitable. For example,
the median functional is elicitable, as minimizing the expected forecasting error with $S(x,y)=|x-y|$ yields the median functional.
If $\rho$ is elicitable, then one can evaluate two point forecasting methods
by comparing their respective expected forecasting error $ES(x, Y)$. As $F_Y$ is unknown, the expected forecasting error can be approximated by the average $\frac{1}{n}\sum_{i=1}^n S(x_i, Y_i)$, where $Y_1, \ldots, Y_n$ are samples that have the distribution $F_Y$ and $x_1, \ldots, x_n$ are the corresponding point forecasts.

If a statistical functional $\rho$ is not elicitable, 
then for any objective function $S$, the minimization of the expected forecasting error does not yield the true value $\rho(F)$. Hence, one cannot tell which one of competing point forecasts for $\rho(F)$ performs the best by comparing their forecasting errors, no matter what objective function $S$ is used.

The concept of elicitability dates back to the pioneering work of \citet{Savage-1971}, \citet{Thomson-1979}, and \citet{Osband-1985} and is further developed by \citet*{Lambert-Pennock-Shoham-2008} and \citet{Gneiting-2011},
who contends that ``in issuing and evaluating point forecasts, it is essential that either the objective function (i.e., the function $S$) be specified ex ante, or
an elicitable target functional be named, such as an expectation
or a quantile, and objective functions be used that are consistent
for the target functional." \citet*{EMW-2009} also points out the critical importance of the specification of an objective function or an elicitable target functional. 


In the present paper, we are concerned with the measurement of risk, which is given by a single-valued statistical functional. Following Definition 2 in \citet{Gneiting-2011}, where the elicitability for a set-valued statistical functional is defined, we define the elicitability for a single-valued statistical functional as follows.\footnote{In Definition \ref{def:elicitable}, the requirement that $\rho(F)$ is the minimum of the set of minimizers of the expected objective function is not essential. In fact, if one replaces the first ``$\min$" in \eqref{equ:e_62} by ``\emph{$\max$}", the conclusions of the paper remain the same;
one only needs to change ``$\VaR_{\alpha}$" to the right quantile $q^+_{\alpha}$ in Theorem \ref{thm:elicitable_measures}.}

\begin{definition}\label{def:elicitable}
  A single-valued statistical functional $\rho(\cdot)$ is elicitable with respect to a class of distributions $\PSet$ if there exists a forecasting objective function $S:\mathbb{R}^2\to\mathbb{R}$ such that \begin{equation}\label{equ:e_62}
  \rho(F)=\min\left\{x\mid x\in \argmin_{x}\int S(x, y)dF(y)\right\},\ \forall F\in\PSet.
  \end{equation}
\end{definition}

In the definition, we only require that $S$ satisfies the condition that $\int S(x, y)dF(y)$ is well-defined and finite for any $F\in\PSet$. We do not need other conditions such as continuity or smoothness on $S$.

\subsection{Main Result}\label{subsec:main_results}

The following Theorem \ref{thm:elicitable_measures} shows that the median shortfall and the mean functional are
the \emph{only} risk measures that
(i) are elicitable; and (ii) have the decision-theoretic foundation of Choquet expected utility (i.e., satisfying Axioms A1-A5). Median shortfall at level $\alpha$ provides a precise description of the average size of loss beyond $\VaR_{\alpha}$ by median; whereas the mean functional captures the tail risk in the sense that knowing $E(L)$ leads to an upper bound $\frac{1}{x}E(L)$ for the tail probability $P(L>x)$ if $L\geq 0$.\footnote{We thank an anonymous referee for pointing this out to us.}


\begin{theorem}\label{thm:elicitable_measures}
Let $\rho:\X\to\mathbb{R}$ be a risk measure that satisfies Axioms A1-A5 and $\X\supset \LInf$. Let $\PSet:=\{F_X\mid X\in \X\}$.
Then, $\rho(\cdot)$ (viewed as a statistical functional on $\PSet$) is elicitable with respect to $\PSet$ if and only if one of the following two cases holds:
  \begin{itemize}
    \item[(i)]
    $\rho=\VaR_{\alpha}$ for some $\alpha\in(0, 1]$ (noting that $\MS_{\alpha}=\VaR_{\frac{\alpha+1}{2}}$ for $\alpha\in[0, 1]$). Here $\VaR_{\alpha}$ is a single valued functional as defined in Section 2.1.
    \item[(ii)]
    $\rho(F)=\int xdF(x)$, $\forall F$.
  \end{itemize}
\end{theorem}
\begin{proof} See Appendix \ref{app:proof_elict_rm}.
\end{proof}

The major difficulty of the proof lies in that the distortion function $h(\cdot)$ in the representation equation \eqref{equ:e_60} of risk measures satisfying Axioms A1-A5 can have various kinds of discontinuities on $[0, 1]$; in particular, the proof is not based on any assumption on left or right continuity of $h(\cdot)$. The outline of the proof is as follows. First, we show that a necessary condition for $\rho$ to be elicitable is that $\rho$ has convex level sets, i.e., $\rho(F_1)=\rho(F_2)$ implies that $\rho(F_1)=\rho(\lambda F_1+(1-\lambda)F_2)$, $\forall \lambda\in(0, 1)$.
The second and the key step is to show that only four kinds of risk measures have convex level sets:
(i) $c\VaR_0+(1-c)\VaR_1$ for some constant $c\in[0, 1]$;
(ii)
$\VaR_{\alpha}$, $\alpha\in(0, 1)$, and, in particular, $\MS_{\alpha}$, $\alpha\in[0, 1)$; (iii)
$\rho=cq_{\alpha}^-+(1-c)q_{\alpha}^+$, where $\alpha\in(0, 1)$ and $c\in[0, 1)$ are constants,  $q_{\alpha}^-(F):=\inf\{x\mid F(x)\geq \alpha\}$, and $q_{\alpha}^+(F):=\inf\{x\mid F(x)> \alpha\}$;
(iv) the mean functional.
Lastly, we examine the elicitability of the aforementioned four kinds of risk measures; in particular, we show that $\rho=cq_{\alpha}^-+(1-c)q_{\alpha}^+$ for $c\in[0,1)$ is not elicitable by extending the main proposition in \citet{Thomson-1979}.


\subsection{Co-elicitability}

The co-elicitability of $k\geq 2$ statistical functionals is a \emph{weaker} notion of elicitability than the notion of elicitability of one statistical functional defined in Definition \ref{def:elicitable}. The notion of co-elicitability is formulated in \citet[][Definition 9]{Lambert-Pennock-Shoham-2008}, which we slightly generalize and rephrase in our notation as follows:\footnote{Without generalization, Definition 9 in \citet{Lambert-Pennock-Shoham-2008} can be rephrased by replacing \eqref{equ:e_98} by:
\begin{equation}
(\rho_1(F), \ldots, \rho_k(F))=\arg \min_{(x_1, \ldots, x_k)}\int S(x_1, \ldots, x_k, y)dF(y),\ \forall F\in\PSet.
\end{equation}
Hence, the only generalization lies in adding $\min\{\cdots\}$ to incorporate the case that there are more than one minimizers to the optimization problem in \eqref{equ:e_98}.}
\begin{definition}\label{def:co_elicitable}
$k\geq 2$ single-valued statistical functionals $\rho_1(\cdot), \ldots, \rho_k(\cdot)$ are called co-elicitable with respect to a class of distributions $\PSet$ if there exists a forecasting objective function $S:\mathbb{R}^{k+1}\to\mathbb{R}$ such that
\begin{align}\label{equ:e_98}
  &(\rho_1(F), \ldots, \rho_k(F))\notag\\
  =&\min\left\{(x_1, \ldots, x_k)\mid (x_1, \ldots, x_k)\in \arg \min_{(x_1, \ldots, x_k)}\int S(x_1, \ldots, x_k, y)dF(y)\right\},\ \forall F\in\PSet.
  \end{align}
\end{definition}

The notion of co-elicitability is weaker than that of elicitability because: (i) if for each $i=1, \ldots, k$, $\rho_i$ is elicitable with a corresponding forecasting objective function $S_i(\cdot, \cdot)$, then $(\rho_1, \ldots, \rho_k)$ are co-elicitable with the corresponding function $S$ being defined as $S(x_1, \ldots, x_k, y):=\sum_{i=1}^k S_i(x_i, y)$; (ii) if $(\rho_1, \ldots, \rho_k)$ are co-elicitable, it does not imply that each $\rho_i$ is elicitable.

\citet{Acerbi-Szekely-2014} show that $(\VaR_{\alpha}, \ES_{\alpha})$ are co-elicitable with respect to a class of distributions $\PSet$ which satisfy some restrictive conditions based on an intuitive argument;
\citet{Fissler-Ziegel-2015} show that $(\VaR_{\alpha}, \ES_{\alpha})$ are co-elicitable with respective to $\PSet=\{F\mid F\ \text{has a continuous density on}\ \mathbb{R}\ \text{and}\ F\ \text{has unique}\ \alpha\ \text{quantile for}\linebreak \text{all}\ \alpha\in(0, 1)\}$, and the corresponding forecasting objective function $S$ in Definition \ref{def:co_elicitable} may be specified as
\begin{align}\label{equ:e_99}
  S(x_1, x_2, y)=&(1_{\{x_1\geq y\}}-\alpha)(-G_1(-x_1) + G_1(- y))+\notag\\
  & \frac{1}{1-\alpha}G_2(-x_2)1_{\{x_1<y\}}(y-x_1)+G_2(-x_2)(x_1-x_2)-\mathcal{G}_2(-x_2),
\end{align}
where $G_1$ and $G_2$ are strictly increasing continuously differentiable functions, $G_1$ is $F$-integrable for any $F\in\PSet$, $\lim_{x\to -\infty} G_2(x)=0$, and $\mathcal{G}_2'=G_2$, e.g., $G_1(x)=x$ and $G_2(x)=e^x$.

The co-elicitability of $(\VaR_{\alpha}, \ES_{\alpha})$ implies that one can evaluate the performance of different forecasting procedures that forecast the \emph{collection} of $(\VaR_{\alpha}, \ES_{\alpha})$ by comparing their realized forecasting errors. More precisely, procedure 1 is considered to better forecast the \emph{collection} of $(\VaR_{\alpha}, \ES_{\alpha})$ than procedure 2 if
\begin{equation}\label{equ:e_84}
\frac{1}{T}\sum_{t=1}^T S(var^1_t, es^1_t, Y_t) < \frac{1}{T}\sum_{t=1}^T S(var^2_t, es^2_t, Y_t),
\end{equation}
where $(var_t^i, es^i_t)$ are the forecasts generated by the $i$th procedure at time $t$, $i=1, 2$, and $Y_t$ is the realized loss at time $t$, $t=1, \ldots, T$.


The co-elicitability of $(\ES_{\alpha}, \VaR_{\alpha})$ does not lead to a reliable method for evaluating forecasts for $\ES_{\alpha}$. More precisely, even if procedure 1 better forecasts the collection $(\VaR_{\alpha}, \ES_{\alpha})$ than procedure 2 in the sense of \eqref{equ:e_84}, procedure 1 may provide much worse forecast of $\ES_{\alpha}$ than procedure 2; this is illustrated in Example \ref{ex:coelicitable_nonsense} and Example \ref{ex:size_deteriorate} at the end of Section \ref{subsec:indirect_backtest}.




Theorem \ref{thm:elicitable_measures} identifies all elicitable risk measures within the class of risk measures that satisfy Axioms A1-A5; a counterpart of the problem studied in Theorem \ref{thm:elicitable_measures} is the following one:
For $k\geq 2$, can we identify all the $k$-tuple of risk measures $(\rho_1, \ldots, \rho_k)$ such that $(\rho_1, \ldots, \rho_k)$ are co-elicitable and each $\rho_i$ satisfies Axioms A1-A5?
Because co-elicitability is weaker than elicitability, the above problem is different from that studied in Theorem \ref{thm:elicitable_measures}; the answer to the problem does not imply Theorem \ref{thm:elicitable_measures}, and Theorem \ref{thm:elicitable_measures} does not provide a complete answer to the problem.

Some examples of risk measures that satisfy the conditions in the above open problem are provided in \citet{Fissler-Ziegel-2015}. In addition to $(\VaR_{\alpha}, \ES_{\alpha})$, $(\VaR_{\alpha_1}, \ldots, \VaR_{\alpha_k}, \sum_{i=1}^k w_i \ES_{\alpha_i})$ are shown to be co-elicitable, where $0<\alpha_1<\cdots<\alpha_k<1$, $(w_1, \ldots, w_k)$ are any weights satisfying $\sum_{i=1}^k w_i=1$ and $w_i>0$, $i=1,\ldots, k$. However, the complete answer to the open problem is not known yet; we leave it for future research.







\subsection{Backtesting a Risk Measure}

As will be shown in the following subsections, there are three approaches for backtesting a risk measure: (i) the direst backtest, which tests if the point estimate or point forecast of the risk measurement under a model is equal to the unknown true risk measurement; (ii) the indirect backtest, which can be classified into two kinds: (a) the first kind of indirect backtests examine if the entire loss distribution,
the entire tail loss distribution, or a collection of statistics including the risk measure of interest under a model are equal to the corresponding
quantities under the true underlying unknown model; (b) the second kind of indirect backtests are based on the co-elicitability of a collection of risk measures; (iii) the forecast evaluation approach based on the elicitability of the risk measure. 

We will also show in the subsections that: (i) VaR and median shortfall can be backtested by all three approaches. (ii) There have been no direct backtesting methods for expected shortfall. (iii) Indirect backtesting methods for expected shortfall have been proposed in the literature. The first kind of indirect backtesting for expected shortfall is a partial backtesting in the sense that: (a) if an indirect backtesting for expected shortfall is not rejected, it will imply that the point forecast for expected shortfall will not be rejected; (b) however, if an indirect backtesting for expected shortfall is rejected, it will be unclear whether the point forecast for expected shortfall should be rejected. The second kind of indirect backtests which are based on the co-elicitability of $(\VaR_{\alpha}, \ES_{\alpha})$ cannot answer the question whether the $\ES_{\alpha}$ forecasted under a bank's model is more accurate that that forecasted under a benchmark model.

\subsubsection{The Direct Backtesting Approach}

The direct backtesting approach is to test whether the risk measurement calculated under a model is equal to the unknown true value of risk measurement. It concerns whether the point estimate or point forecast of the risk measure is acceptable or not. For example, suppose a bank reports that the $\VaR_{99\%}$ of its trading book is 1 billion. The direct backtesting approach answers the question whether the single number 1 billion is acceptable or not.

More precisely, suppose the loss of a bank on the $t$th day is $L_t$, $t=1, 2, \ldots, T$. On each day $t-1$, the bank forecasts the risk measurement $\rho$
of $L_t$ based on the information available on day $t-1$, which is denoted as $\mathcal{F}_{t-1}$. Let $G_{t|t-1}$ denote the bank's model of the
conditional distribution of $L_t$ given $\mathcal{F}_{t-1}$, and let
$\rho^{G_{t|t-1}}(L_t)$ denote the risk measurement of $L_t$ under the model $G_{t|t-1}$. Suppose the unknown true conditional distribution of $L_{t}$ given $\mathcal{F}_{t-1}$ is $F_{t|t-1}$ and the true risk measurement is denoted as $\rho^{F_{t|t-1}}(L_t)$. Then, the direst backtesting of the risk measure $\rho$ is to test
\begin{equation}\label{equ:backtest_rho_cond}
H_0: \rho^{G_{t|t-1}}(L_t) = \rho^{F_{t|t-1}}(L_t),\ \forall t=1, \ldots, T; \ H_1: \text{otherwise}.
\end{equation}
For $\rho=\VaR_{\alpha}$, the null hypothesis in \eqref{equ:backtest_rho_cond} is equivalent to that $I_t:=1_{\{L_t > \VaR_{\alpha}(L_t)\}}$, $t=1, \ldots, T$, are i.i.d. Bernoulli($1-\alpha$) random variables (\citet{Christoffersen-1998}, Lemma 1). Based on such observation, \citet{Kupiec-1995} propose the proportion of failure test for backtesting VaR, which is closely related to the ``traffic light" approach of backtesting VaR adopted in the Basel Accord \citep{Basel-Amd-1996, Basel06}. \citet{Christoffersen-1998} propose conditional coverage and independence tests for VaR within a first-order Markov process model. For more recent development on the backtesting of VaR, see \citet{Lopez-1999a}, \citet{Lopez-1999b}, \citet{Engle-Manganelli-2004}, \citet{Chris-Pelletier-2004}, \citet{Haas-2005}, \citet{Campbell-2006}, \citet*{Christoffersen-2010}, \citet*{BCP-2011}, \citet*{Linton-2011}, etc.

As $\MS_{\alpha}=\VaR_{(1+\alpha)/2}$, the backtesting of median shortfall is exactly the same as that of VaR. In contrast, \emph{there have been no direct backtesting methods for expected shortfall} in the existing literature. The reason might be simple: The null hypothesis for direct backtesting expected shortfall is that $\ES_{\alpha}^{G_{t|t-1}}(L_t) = \ES_{\alpha}^{F_{t|t-1}}(L_t)$. It might be \emph{difficult (if not impossible)} to find a statistic whose distribution is known under the null hypothesis. In contrast, the distribution of the indicator random variable $I_t=1_{\{L_t > \VaR_{\alpha}(L_t)\}}$ is known under the null hypothesis for direct backtesting VaR, and hence $I_t$ can be used to construct test statistic for direct backtesting VaR.

\subsubsection{The Indirect Backtesting Approach}\label{subsec:indirect_backtest}

There are two kinds of indirect backtesting approaches. The first kind of indirect backtesting approach concerns whether the bank's model of the \emph{entire loss distribution} is the same as the unknown true loss distribution. More precisely, the indirect backtesting approach is to test:
\begin{equation}\label{equ:backtest_dist_cond}
H_0: G_{t|t-1}(x) = F_{t|t-1}(x),\ \forall x\in \mathbb{R},\ \forall t=1, \ldots, T; \ H_1: \text{otherwise}.
\end{equation}
If the null hypothesis is not rejected, then it will imply that $\rho^{G_{t|t-1}}(L_t) = \rho^{F_{t|t-1}}(L_t)$, i.e., the risk measurement will not be rejected; however, if the null hypothesis is rejected, then it will be unclear whether the point forecast $\rho^{G_{t|t-1}}(L_t)$ should be rejected or not. Therefore, the kind of indirect backtesting approach can only serve as a \emph{partial} backtesting of a particular risk measure. For example, suppose a bank reports that the $\ES_{99\%}$ of its trading book is 1 billion. Using the indirect backtesting approach, one can test the bank's model of the entire loss distribution. If the test is not rejected, then it will imply that the number 1 billion is acceptable; however, if the test is rejected, then it will be unclear if the number 1 billion should be accepted or rejected.

Strictly speaking, this indirect backtesting approach shall not be regarded as an approach for backtesting a particular risk measure, because the backtesting has nothing to do with any particular risk measure, although the test has partial implication on the acceptability of the point forecast of a particular risk measure.

This kind of indirect backtesting approaches have been proposed for backtesting expected shortfall in the literature. \citet{Berkowitz-2001} propose likelihood ratio tests based on censored Gaussian likelihood for the test \eqref{equ:backtest_dist_cond}. \citet{Kerkhof-Melenberg-2004} propose a functional delta method for testing the hypothesis \eqref{equ:backtest_dist_cond}.
 \citet{Acerbi-Szekely-2014} propose three indirect tests for backtesting $\ES_{\alpha}$. The first two tests are to test the \emph{entire tail loss distribution} under the assumption that $\VaR_{\alpha}$ has already been tested and that $L_1, \ldots, L_T$ are independent:
\begin{equation}\label{equ:backtest_dist_AS_1}
H_0: G_{t|t-1, \alpha}(x) = F_{t|t-1, \alpha}(x),\ \forall x\in \mathbb{R},\ \forall t=1, \ldots, T; \ H_1: \text{otherwise},
\end{equation}
where $G_{t|t-1, \alpha}$ and $F_{t|t-1, , \alpha}$ and the $\alpha$-tail distribution of $G_{t|t-1}$ and $F_{t|t-1}$ respectively (see Example \ref{ex:es} for definition of $\alpha$-tail distribution). The third test is the same as the test \eqref{equ:backtest_dist_cond}. All the three tests proposed by the authors require that one knows how to simulate random samples with distribution $G_{t|t-1}(\cdot)$ in order to simulate the test statistic and to calculate the $p$ value of the test.
\citet{Costanzino-Curran-2015} propose an approach to indirectly backtest $\ES_{\alpha}$ by testing:
\begin{align}\label{equ:backtest_tail_quantile}
H_0: & \int_{\alpha}^11_{\{L_t\leq \VaR_{p}(L_t)\}}dp, t=1, \dots, T, \text{are i.i.d.,}\ \VaR^{F_{t|t-1}}_p(L_t) = \VaR^{G_{t|t-1}}_{p}(L_t),\notag\\
&  \forall p\in [\alpha, 1), t=1, \ldots, T\notag\\
H_1: & \text{otherwise}.
\end{align}
This approach does not need to simulate random samples under the null hypothesis in order to calculate the $p$ value. \citet{McNeil-Frey-2000} assume that the loss process $\{L_t, t=1, \ldots, T\}$ follows the dynamics $L_t = m_t + s_t Z_t$, where $m_t$ and $s_t$ are respectively the conditional mean and conditional standard deviation, and $Z_t$ is a strict white noise. Under this assumption, they propose to backtest $\ES_{\alpha}$ by testing the hypothesis:
\begin{align}\label{equ:hypo_McNeil_Frey}
H_0: & m^{G_{t|t-1}}_t = m_t, s^{G_{t|t-1}}_t = s_t, \VaR^{G_{t|t-1}}_{\alpha}(L_t)=\VaR^{F_{t|t-1}}_{\alpha}(L_t),\notag\\ & \ES^{G_{t|t-1}}_{\alpha}(L_t)=\ES^{F_{t|t-1}}_{\alpha}(L_t), \forall t;\notag\\
H_1: & \text{otherwise}.
\end{align}
This test is an indirect test for $\ES_{\alpha}$ because if the null hypothesis is rejected, it is not clear if the claim $\ES^{G_{t|t-1}}_{\alpha}(L_t)=\ES^{F_{t|t-1}}_{\alpha}(L_t), \forall t$ should be rejected or not.

The second kind of indirect backtests are those based on the co-elicitability of a collection of risk measures. For example, let $(\VaR_{\alpha}^{Ben}(L_t), \ES_{\alpha}^{Ben}(L_t))$, $t=1, \ldots, T\}$, be the $(\VaR_{\alpha},\ES_{\alpha})$ forecasted under a benchmark model such as a standard model specified by the regulator. \citet{Fissler-Ziegel-Gneiting-2015} propose the following two indirect backtests for backtesting $\ES_{\alpha}$:
\begin{align}\label{equ:backtest_var_es_co}
H_0^-: & E_{t-1}[S(\VaR_{\alpha}^{G_{t|t-1}}(L_t), \ES_{\alpha}^{G_{t|t-1}}(L_t), L_t)] \geq E_{t-1}[S(\VaR_{\alpha}^{Ben}(L_t), \ES_{\alpha}^{Ben}(L_t), L_t)], \forall t\notag\\
H_1^-: & \text{otherwise};\notag\\
H_0^+: & E_{t-1}[S(\VaR_{\alpha}^{G_{t|t-1}}(L_t), \ES_{\alpha}^{G_{t|t-1}}(L_t), L_t)] \leq E_{t-1}[S(\VaR_{\alpha}^{Ben}(L_t), \ES_{\alpha}^{Ben}(L_t), L_t)], \forall t\notag\\
H_1^+: & \text{otherwise},
\end{align}
where $S(\cdot, \cdot, \cdot)$ is the forecasting objective function defined in \eqref{equ:e_99} with $G_1(x)=x$ and $G_2(x)=e^x/(1+e^x)$.

These tests are indirect backtests for $\ES_{\alpha}$ because no matter these tests are rejected or not, we do now know whether $\ES_{\alpha}^{G_{t|t-1}}$ is more accurate than $\ES_{\alpha}^{Ben}(L_t)$.
In fact, these tests are not able to find out which model gives a more accurate forecast for $\ES_{\alpha}$, as is shown in Example \ref{ex:coelicitable_nonsense}.

\begin{example}\label{ex:coelicitable_nonsense}
Suppose the true distribution of a bank's loss random variable $L$ is $N(\mu, \sigma^2)$ with $\mu=-1.5$, $\sigma=1.0$. Let $\alpha=0.975$, which is suggested in \citet{Basel-2013}. Then the true value of $(\VaR_{\alpha}(L), \ES_{\alpha}(L))$ is $(VaR_{\alpha}, ES_{\alpha})=(0.460, 0.838)$. Suppose the forecasts given by a bank's model are $(VaR_{\alpha}, x \cdot ES_{\alpha})$ and those given by a benchmark model (prefered by the regulator) are $(x \cdot VaR_{\alpha}, ES_{\alpha})$, where $0 < x < 1$; hence, the bank's model always under-forecasts $\ES_{\alpha}$ but the benchmark model always truthfully forecasts $\ES_{\alpha}$; therefore, the bank's model should be rejected. However, these tests will conclude that the bank's model are better than the benchmark model because the forecasting error of the bank's model (i.e., $E[S(VaR_{\alpha}, x \cdot ES_{\alpha}, L)]$) is always smaller than that of the benchmark model (i.e., $E[S(x \cdot VaR_{\alpha}, ES_{\alpha}, L)]$) for any $x\in (0.55, 1.0)$. In other words, even if the bank's model {\it under-forecasts} the $\ES_{\alpha}$ by as much as 45\%, it will still be \emph{wrongly} considered to be better than the benchmark model that truthfully forecasts $\ES_{\alpha}$, \emph{mainly due to the fact that co-elicitability does not imply elicitability, and some rather strange behavior of the forecasting objective function $S$ defined in \eqref{equ:e_99}}. This is illustrated by Figure \ref{fig:nonsense_coelicitability}.
\begin{figure}[htb]
\centering
  \includegraphics[keepaspectratio=true, width=\textwidth]{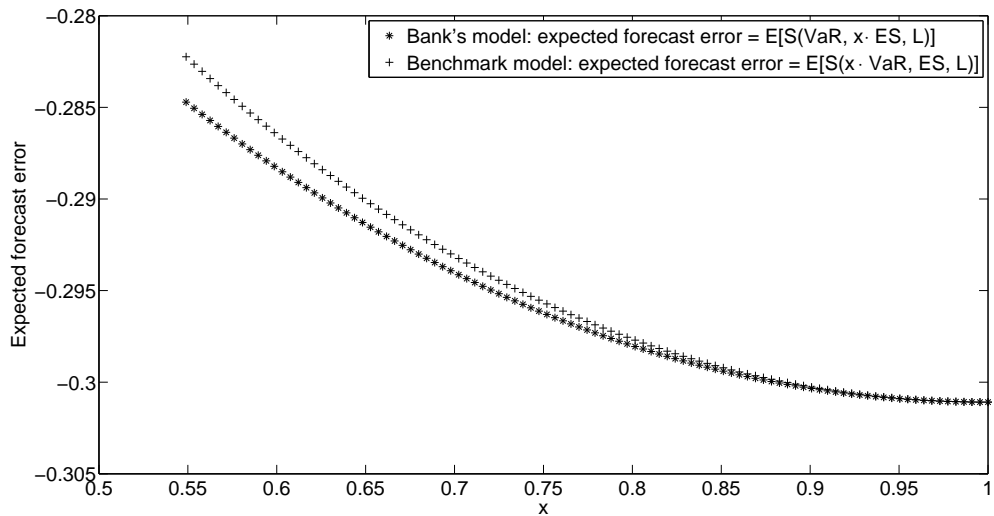}
  \caption{A graph for the counterexample in Example \ref{ex:coelicitable_nonsense}. The forecasting error of the bank's model (i.e., $E[S(VaR_{\alpha}, x \cdot ES_{\alpha}, L)]$) is always smaller than that of the benchmark model (i.e., $E[S(x \cdot VaR_{\alpha}, ES_{\alpha}, L)]$) for any $x\in (0.55, 1.0)$; therefore, the tests in \eqref{equ:backtest_var_es_co} will conclude that the bank's model better forecasts $\ES_{\alpha}$ than the benchmark model. However, the bank's model always under-forecasts $\ES_{\alpha}$, while the benchmark model always truthfully forecasts $\ES_{\alpha}$. Such inconsistency, mainly due to the fact that co-elicitability does not imply elicitability, shows that the tests in \eqref{equ:backtest_var_es_co} are not able to find out which model gives a more accurate forecast for $\ES_{\alpha}$.}
  \label{fig:nonsense_coelicitability}
\end{figure}
\end{example}


Another drawback of these backtests is that the performance of the backtests further deteriorates when the scale of the loss random variable increases, because the term $G_2(-x_2)$ in Eq. \eqref{equ:e_99} goes to zero as $x_2$ goes to infinity. The consequence is that larger banks can more easily under-report $\ES$ than smaller banks if such backtests are used for backtesting $\ES_{\alpha}$. This is illustrated in Example \ref{ex:size_deteriorate}.

\begin{example}\label{ex:size_deteriorate}
Suppose there is a larger bank whose loss random variable is $15$ times of the loss $L$ in Example \ref{ex:coelicitable_nonsense}. Thus, the loss random variable of this larger bank has a normal distribution $N(\mu, \sigma^2)$ with $\mu=-1.5\times 15$, $\sigma=15.0$. Let $\alpha=0.975$. Note the true value of $(\VaR_{\alpha}, \ES_{\alpha})$ is $(VaR_{\alpha}, ES_{\alpha})=(0.460, 0.838)\times 15$. Suppose the forecasts given by a bank's model are $(VaR_{\alpha}, x \cdot ES_{\alpha})$ and those given by a benchmark model (prefered by the regulator) are $(x \cdot VaR_{\alpha}, ES_{\alpha})$. Again, as in Figure \ref{fig:nonsense_coelicitability}, Figure \ref{fig:size_deteriorate} shows that the backtests make the wrong conclusion on which model better forecasts $\ES_{\alpha}$. In addition, Figure \ref{fig:size_deteriorate} shows that the forecasting error for the bank's model almost remain unchanged when $x\in (0.55, 1.0)$, which is due to the fact that when $ES_{\alpha}$ is large enough, the term
$E[\frac{1}{1-\alpha}G_2(-x\cdot ES_{\alpha})1_{\{VaR_{\alpha}<L\}}(L-VaR_{\alpha})+G_2(-x \cdot ES_{\alpha})(VaR_{\alpha}-x ES_{\alpha})-\mathcal{G}_2(-x \cdot ES_{\alpha})]$
in the expected forecasting error will be so small that the expected forecasting error will not change much when $x$ varies. In other words, when the scale of the loss random variable $L$ is large enough, the expected forecasting error $E[S(VaR_{\alpha}, x \cdot ES_{\alpha}, L)]$ becomes insensitive to the value of $x$. This counterexample happens again \emph{mainly due to some strange behavior of the forecasting objective function $S$ defined in \eqref{equ:e_99}}.

\begin{figure}[htb]
\centering
  \includegraphics[keepaspectratio=true, width=\textwidth]{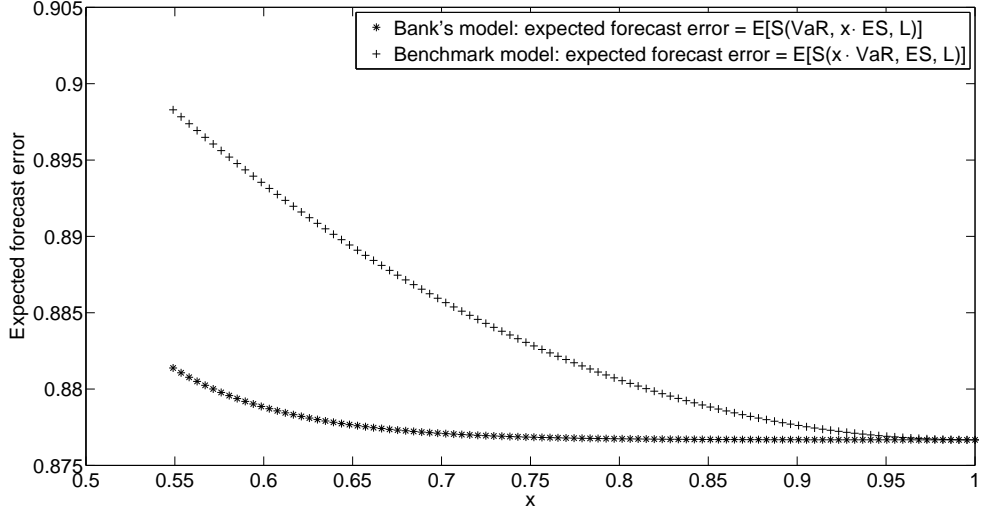}
  \caption{A graph for the counterexample in Example \ref{ex:size_deteriorate}. The expected forecasting error of the bank's model (i.e., $E[S(VaR_{\alpha}, x \cdot ES_{\alpha}, L)]$) in Example \ref{ex:size_deteriorate}
   almost remain unchanged when $x\in (0.55, 1.0)$, because when $ES_{\alpha}$ is large enough, the term $E[\frac{1}{1-\alpha}G_2(-x\cdot ES_{\alpha})1_{\{VaR_{\alpha}<L\}}(L-VaR_{\alpha})+G_2(-x \cdot ES_{\alpha})(VaR_{\alpha}-x ES_{\alpha})-\mathcal{G}_2(-x \cdot ES_{\alpha})]$ in the expected forecasting error will be so small that the expected forecasting error will not change much when $x$ varies. In other words, when the scale of the loss random variable $L$ is large enough, the expected forecasting error $E[S(VaR_{\alpha}, x \cdot ES_{\alpha}, L)]$ becomes insensitive to the value of $x$.}
   \label{fig:size_deteriorate}
\end{figure}
\end{example}

\subsubsection{The Backtesting Approach Based on the Elicitability of a Risk Measure}

The backtesting approach based on the forecast evaluation framework and elicitability has been proposed to backtest VaR.
This approach requires a benchmark model because the elicitability concerns the comparison of multiple models rather than the validation of a single model. \citet{Lopez-1999a} propose to define
the forecasting error for $\VaR_{\alpha}$ under the model $G_{t|t-1}$ as $\sum_{t=1}^T S(\VaR_{\alpha}^{G_{t|t-1}}(L_t), L_t)$, where
$S(\cdot, \cdot)$ is a forecast objective function (loss function). Since $\VaR_{\alpha}$ is elicitable, $S$ can be defined as $S_{\alpha}(x, y)=(1_{\{x\geq y\}} - \alpha)(x - y)$.
Then, the forecasting error is compared with a benchmark forecasting error calculated under a benchmark model to backtest $\VaR_{\alpha}$.

In contrast, expected shortfall cannot be backtested by this approach because it is not elicitable, and therefore, no function $S$ can be used to define the forecasting error.





\section{Extension to Incorporate Multiple Models}\label{sec:mult_scen}

The previous section address the issue of model uncertainty from the perspective of elicitability. Following \citet{Gilboa-Schmeidler-1989} and \citet{Hansen-Sargent-2001, HS-2007}, we further incorporate robustness by considering multiple models (scenarios).
More precisely, we consider $m$ probability measures $P_i$, $i=1, \ldots, m$ on the state space $(\Omega, \mathcal{F})$. Each $P_i$ corresponds to one model or one scenario, which may refer to a specific economic regime such as an economic boom and a financial crisis. The loss distribution of a random loss $X$
under different scenarios can be substantially different. For example, the VaR calculated under the scenario of the 2007 financial crisis is much higher than that under a scenario corresponding to a normal market condition due to the difference of loss distributions.

Suppose that under the $i$th scenario, the measurement of risk is given by $\rho_i$ that satisfy Axioms A1-A5. Then by Lemma \ref{lemma:risk_measure_Schmeidler}, $\rho_i$ can be represented by $\rho_i(X)=\int Xd(h_i\circ P_i)$, where $h_i$ is a distortion function, $i=1, \ldots, m$. We then propose the following risk measure to incorporate multiple scenarios:
\begin{equation}\label{equ:e_79}
\rho(X)=f(\rho_1(X), \rho_2(X), \ldots, \rho_m(X)),
\end{equation}
where $f:\mathbb{R}^m\to\mathbb{R}$ is called a scenario aggregation function.

We postulate that the scenario aggregation function $f$ satisfies the following axioms:

\noindent\textbf{Axiom B1.} Positive homogeneity and translation
scaling: $f(a\tilde{x}+b\mathbf{1})=af(\tilde{x})+sb,\ \forall \tilde{%
x}\in \mathbb{R}^{m},\forall a\geq 0,\forall b\in \mathbb{R}$, where $s>0$ is a constant and $\mathbf{1}:=(1,1,...,1)\in
\mathbb{R}^{m}$.

\noindent\textbf{Axiom B2.} Monotonicity: $f(\tilde{x})\leq f(%
\tilde{y})$, if $\tilde{x}\leq \tilde{y}$, where $\tilde{x}\leq
\tilde{y}$ means $x_i\leq y_i, i=1, \ldots, m$.

\noindent\textbf{Axiom B3.} Uncertainty aversion: if $f(\tilde x)=f(\tilde y)$, then for any $\alpha\in(0, 1)$, $f(\alpha\tilde x+(1-\alpha)\tilde y)\leq f(\tilde x)$.

Axiom B1 states that if the risk measurement of $Y$ is an affine function of that of $X$ under each scenario, then the aggregate risk measurement of $Y$ is also an affine function of that of $X$. Axiom B2 states that if the risk measurement of $X$ is less than or equal to that of $Y$ under each scenario, then the aggregate risk measurement of $X$ is also less than or equal to that of $Y$. 
Axiom B3 is proposed by \citet{Gilboa-Schmeidler-1989} to ``capture the phenomenon of hedging"; it is used as one of the axioms for the maxmin expected utility that incorporates robustness.

\begin{lemma}\label{lemma:scenario_agg}
  A scenario aggregation function $f:\mathbb{R}^m\to\mathbb{R}$ satisfies Axioms B1-B3 if and only if there exists a
  set of
weights $\mathcal{W}=\{\tilde w\}\subset \mathbb{R}^{m}$ with each
$\tilde w=(w_1,\ldots,
w_m)\in \mathcal{W}$ satisfying $w_i\geq 0$ and
$\sum_{i=1}^{m}w_i=1$, such that
  \begin{equation}\label{equ:e_81}
  f(\tilde x)= s\cdot\sup_{\tilde w\in \mathcal{W}}\left\{\sum_{i=1}^{m}w_ix_i\right\}, \forall \tilde x\in\mathbb{R}^m.
  \end{equation}
\end{lemma}
\begin{proof} First, we show that Axioms B1-B3 are equivalent to the Axioms C1-C4 in \citet*{KPH-2013} with $n_i=1$, $i=1, \ldots, m$. Axioms B1 and B2 are the same as the Axioms C1 and C2, respectively. Axiom C4 holds for any function when $n_i=1$, $i=1, \ldots, m$. Axioms C1 and C3 apparently implies Axiom B3. We will then show that Axiom B1 and B3 imply Axiom C3. In fact, For any $\tilde x$ and $\tilde y$, it follows from Axiom B1 that $f(\tilde x-f(\tilde x)/s)=f(\tilde y-f(\tilde y)/s)=0$. Then, it follows from Axioms B1 and B3 that
$f(\tilde x+\tilde y)-f(\tilde x)-f(\tilde y)=f(\tilde x -f(\tilde x)/s +\tilde y-f(\tilde y)/s)=2f(\frac{1}{2}(\tilde x -f(\tilde x)/s) +\frac{1}{2}(\tilde y-f(\tilde y)/s))\leq 2f(\tilde x -f(\tilde x)/s)=0$.
Hence, Axiom C3 holds. Therefore, Axioms B1-B3 are equivalent to Axioms C1-C4, and hence the conclusion of the lemma follows from Theorem 3.1 in \citet*{KPH-2013}.
\end{proof}
In the representation \eqref{equ:e_81},
each weight $\tilde w\in\mathcal{W}$ can be regarded as a prior probability on the set of scenarios; more precisely, $w_i$ can be viewed as the likelihood that the scenario $i$ happens.

Lemma \ref{lemma:risk_measure_Schmeidler} and Lemma \ref{lemma:scenario_agg} lead to the following class of risk measures:\footnote{\citet{Gilboa-Schmeidler-1989} consider $\inf_{P\in \mathcal{P}}\int u(X)\,dP$ without $h_i$; see also \citet{Xia-2013}.}
\begin{equation}\label{equ:e_83}
\rho(X)=s\cdot\sup_{\tilde w\in \mathcal{W}}\left\{\sum_{i=1}^{m}w_i\int X\,d(h_i\circ P_i)\right\}.
\end{equation}
By Theorem \ref{sec:main_results}, the requirement of elicitability under each scenario leads to the following tail risk measure
\begin{equation}\label{equ:e_82}
\rho(X)=s\cdot\sup_{\tilde w\in \mathcal{W}}\left\{\sum_{i=1}^{m}w_i\MS_{i,\alpha_i}(X)\right\},
\end{equation}
where $\MS_{i,\alpha_i}(X)$ is the median shortfall of $X$ at confidence level $\alpha_i$ calculated under the $i$th scenario (model). The risk measure $\rho$ in \eqref{equ:e_82}
addresses the issue of model uncertainty and incorporate robustness from two aspects: (i) under each scenario $i$, $\MS_{i,\alpha_i}$ is elicitable and statistically robust (\citet*{KPH-2006, KPH-2013} and \citet*{CDS-2010}); (ii) $\rho$ incorporates multiple scenarios and multiple priors on the set of scenarios.

\section{Application to Basel Accord Capital Rule for Trading Books}\label{sec:app}

What risk measure should be used for setting capital requirements for banks is an important issue that has been under debate since the 2007 financial crisis. The Basel II use a 99.9\% VaR for setting capital requirements for banking books of financial institutions (\citet{Gordy-2003}).
The Basel II capital charge for the trading book on the $t$th day is specified as
$\rho_t(X_t, X_{t-1}, \ldots, X_{t-59}):=s_t\max\left\{\frac{1}{s_t}\VaR_{t-1}(X_t), \frac{1}{60}\sum_{i=1}^{60}
\VaR_{t-i}(X_{t-i+1})\right\}$, where $X_{t-i}$ is the trading book loss on the $(t-i)$th day; $s_t\geq 3$ is a constant that is specified by the regulator based on the backtesting result of the institution's VaR model;
$\VaR_{t-i}(X_{t-i+1})$ is the 10-day VaR at 99\% confidence level calculated on day $t-i$, which corresponds to the $i$th model, $i=1,\ldots,60$. Define the 61th model under which $X=0$ with probability one. Assume that the trading book composition and the size of the positions remain the same over the 60 day periods. Then, $X_t, X_{t-1}, \ldots, X_{t-59}$ can be regarded as the realization of the same random loss under different distributions. In such case, the Basel II risk measure is a special case of the class of risk measures considered in \eqref{equ:e_82}; it
incorporates 61 models and two priors: one is $\tilde w=(1/s, 0, \ldots, 0, 1-1/s)$, the other $\tilde w=(1/60, 1/60, \ldots, 1/60, 0)$. The Basel 2.5 risk measure (\citet{Basel09}) mitigates the procyclicality of the Basel II risk measure by incorporating the ``stressed VaR" calculated under stressed market conditions such as financial crisis. The Basel 2.5 risk measure can also be written in the form of \eqref{equ:e_82}. 

In a consultative document released by the Bank for International Settlement (\citet{Basel-2013}), the Basel Committee proposes to ``move from value-at-risk to expected shortfall," which ``measures the riskiness of a position by considering both the size and the likelihood of losses above a certain confidence level." The proposed new Basel (called Basel 3.5) capital charge for the trading book measured on the $t$th day is defined as
$\rho_t(X_t, X_{t-1}, \ldots, X_{t-59}):=s\max\left\{\frac{1}{s}\ES_{t-1}(X_t), \frac{1}{60}\sum_{i=1}^{60}
\ES_{t-i}(X_{t-i+1})\right\}$, where $\ES_{t-i}(X_{t-i+1})$ is the ES at 97.5\% confidence level calculated on day $t-i$, $i=1,\ldots,60$. Assume that the trading book composition and the size of the positions remain the same over the 60 day periods. Then, the proposed Basel 3.5 risk measure is a special case of the class of risk measures considered in \eqref{equ:e_83}.\footnote{The Basel II, Basel 2.5, and newly proposed risk measure (Basel 3.5) for the trading book are also special cases of the class of risk measures called natural risk statistics proposed by \citet*{KPH-2013}. The natural risk statistics are axiomatized by a different set of axioms including a comonotonic subadditivity axiom.}

The major argument for the change from VaR to ES is that ES better captures tail risk than VaR. The statement that the 99\% VaR is 100 million dollars does no carry information as to the size of loss in cases when the loss does exceed 100 million; on the other hand, the 99\% ES measures the mean of the size of loss given that the loss exceeds the 99\% VaR.

Although the argument sounds reasonable, ES is not the only risk measure that captures tail risk; in particular, an alternative risk measure that captures tail risk is median shortfall (MS), which, in contrast to expected shortfall, measures the median rather than the mean of the tail loss distribution. For instance, in the aforementioned example, if we want to capture the size and likelihood of loss beyond the 99\% VaR level, we can use either ES at 99\% level, or, alternatively, MS at 99\% level.

MS may be preferable than ES for setting capital requirements in banking regulation because (i) MS is elicitable but ES is not; and (ii) MS is robust but ES is not (\citet*{KPH-2006, KPH-2013} and \citet*{CDS-2010}). \citet*{KPH-2013}
show that robustness is indispensable for external risk measures used for legal enforcement such as calculating capital requirements.

To further compare the robustness of MS with ES, we carry out a simple empirical study on the measurement of tail risk of S\&P 500 daily return. We consider two IGARCH(1, 1) models similar to the model of RiskMetrics:
\begin{itemize}
  \item
      Model 1: IGARCH(1, 1) with conditional distribution being Gaussian
    \begin{equation*}
    r_t=\mu+\sigma_t\epsilon_t,\ \sigma_t^2=\beta\sigma_{t-1}^2+(1-\beta)r_{t-1}^2,
    \epsilon_t \overset{d}{\sim}N(0, 1).
    \end{equation*}
 \item
     Model 2: the same as model 1 except that the conditional distribution is specified as $\epsilon_t \overset{d}{\sim}t_{\nu}$, where $t_{\nu}$ denotes $t$ distribution with degree of freedom $\nu$.
\end{itemize}
We respectively fit the two models to the historical data of daily returns of S\&P 500 Index during 1/2/1980--11/26/2012 and then forecast the one-day MS and ES
of a portfolio of S\&P 500 stocks that is worth 1,000,000 dollars on 11/26/2012.
The comparison of the forecasts of MS and ES under the two models is shown in Table \ref{tab:MS_ES_SP}, where
$\ES_{\alpha,i}$ and $\MS_{\alpha,i}$ are the $\ES_{\alpha}$ and $\MS_{\alpha}$ calculated under the $i$th model, respectively, $i=1, 2$. It is clear from the table that the change of ES under the two models (i.e., $\ES_{\alpha,2}-\ES_{\alpha,1}$) is much larger than that of MS (i.e., $\MS_{\alpha,2}-\MS_{\alpha,1}$), indicating that ES is more sensitive to model misspecification than MS.

\begin{table}[htbp]
  \centering
  \caption{The comparison of the forecasts of one-day MS and ES of a portfolio of S\&P 500 stocks that is worth 1,000,000 dollars on 11/26/2012. $\ES_{\alpha,i}$ and $\MS_{\alpha,i}$ are the ES and MS at level $\alpha$ calculated under the $i$th model, respectively, $i=1, 2$. It is clear that the change of ES under the two models (i.e., $\ES_{\alpha,2}-\ES_{\alpha,1}$) is much larger than that of MS (i.e., $\MS_{\alpha,2}-\MS_{\alpha,1}$).}
  \resizebox{\linewidth}{!}{
    \begin{tabular}{cccccccc}
    \addlinespace
    \toprule
    \multirow{2}[0]{*}{$\alpha$} & \multicolumn{3}{c}{ES} & \multicolumn{3}{c}{MS} & \multirow{2}[0]{*}{$\frac{\ES_{\alpha,2}-\ES_{\alpha,1}}{\MS_{\alpha,2}-\MS_{\alpha,1}}-1$} \\
    \cline{2-7}
    \addlinespace
          & $\ES_{\alpha,1}$ & $\ES_{\alpha,2}$ & $\ES_{\alpha,2}-\ES_{\alpha,1}$ & $\MS_{\alpha,1}$ & $\MS_{\alpha,2}$ & $\MS_{\alpha,2}-\MS_{\alpha,1}$ &  \\
    97.0\% & 19956 & 21699 & 1743  & 19070 & 19868 & 798   & 118.4\% \\
    97.5\% & 20586 & 22690 & 2104  & 19715 & 20826 & 1111  & 89.3\% \\
    98.0\% & 21337 & 23918 & 2581  & 20483 & 22011 & 1529  & 68.8\% \\
    98.5\% & 22275 & 25530 & 3254  & 21441 & 23564 & 2123  & 53.3\% \\
    99.0\% & 23546 & 27863 & 4317  & 22738 & 25807 & 3070  & 40.6\% \\
    99.5\% & 25595 & 32049 & 6454  & 24827 & 29823 & 4996  & 29.2\% \\
    \bottomrule
    \end{tabular}%
    }
  \label{tab:MS_ES_SP}%
\end{table}%



\section{Comments}\label{sec:comm}

\subsection{Criticism of Value-at-Risk}

As pointed out by \citet[][]{Aumann-Serrano-2008}, ``like any index or summary statistic, \dots, the riskiness
index summarizes a complex, high-dimensional object by a single number. Needless to say, no index captures all the relevant aspects of the
situation being summarized." Below are some popular criticisms of VaR in the literature.

(i) The VaR at level $\alpha$ does not provide information regarding the size of the tail loss distribution beyond $\VaR_{\alpha}$. However, the median shortfall at level $\alpha$ does address this issue by measuring the median size of the tail loss distribution beyond $\VaR_{\alpha}$.

(ii) There is a  pathological counterexample that, for some level $\alpha$, the $\VaR_{\alpha}$ of a fully concentrated portfolio might be smaller than that of a fully diversified portfolio, which is against the economic intuition that diversification reduces risk; see Example 6.7 in \citet[][p. 241]{MFE2005}.
However, this counterexample disappears if $\alpha > 98\%$.

(iii) VaR does not satisfy the mathematical axiom of subadditivity (\citet{Huber-1981}, \citet{Artz99})\footnote{The representation theorem in \citet{Artz99} is based on  \cite{Huber-1981}, who use the same set of axioms. \citet{Gilboa-Schmeidler-1989} obtains a more general representation based on a different set of axioms.}. However, the subadditivity axiom is somewhat controversial: (1) The subadditivity axiom is based on an intuition that ``a merger does not create extra risk" (\citet{Artz99}, p. 209), which may not be true, as can be seen from the merger of Bank
of America and Merrill Lynch in 2008.
(2) Subadditivity is related to the idea that diversification is beneficial; however, diversification may not always be beneficial. \citet[][pp. 271--272]{Fama-Miller-1972} show that diversification is ineffective for asset returns with heavy tails (with tail index less than 1); these results are extended in  \citet{Ibragimov-Walden-2007} and \citet{Ibragimov-2009}. See \citet*[][Sec. 6.1]{KPH-2013} for more discussion.
(3) Although subadditivity ensures that $\rho(X_1)+\rho(X_2)$ is an upper bound for $\rho(X_1+X_2)$, this upper bound may not be valid in face of model uncertainty.\footnote{In fact, suppose we are concerned with obtaining an upper bound for $\ES_{\alpha}(X_1+X_2)$. In practice, due to model uncertainty, we can only compute $\widehat{\ES}_{\alpha}(X_1)$ and $\widehat{\ES}_{\alpha}(X_2)$, which are estimates of $\ES_{\alpha}(X_1)$ and $\ES_{\alpha}(X_2)$ respectively. $\widehat{\ES}_{\alpha}(X_1)+\widehat{\ES}_{\alpha}(X_2)$ cannot be used as an upper bound for $\ES_{\alpha}(X_1+X_2)$ because it is possible that $\widehat{\ES}_{\alpha}(X_1)+\widehat{\ES}_{\alpha}(X_2)< \ES_{\alpha}(X_1)+\ES_{\alpha}(X_2)$.}
(4) In practice, $\rho(X_1)+\rho(X_2)$ may not be a useful upper bound for $\rho(X_1+X_2)$ as the former may be too much larger than the latter.\footnote{For example,
let $X_1$ be the loss of a long position of a call option on a stock (whose price is \$100) at strike \$100 and let $X_2$ be the loss of a short position of a call option on that stock at strick \$95. Then the margin requirement for $X_1+X_2$, $\rho(X_1+X_2)$, should not be larger than \$5, as $X_1+X_2\leq 5$. However, $\rho(X_1)=0$ and $\rho(X_2)\approx 20$ (the margin is around 20\% of the underlying stock price). In this case, no one would use the subadditivity to charge the upper bound $\rho(X_1)+\rho(X_2)\approx 20$ as the margin for the portfolio $X_1+X_2$; instead, people will directly compute $\rho(X_1+X_2)$.} (5) Subadditivity is not necessarily needed for capital allocation or asset allocation.\footnote{
\citet*[][Sec. 7]{KPH-2013}
derive the Euler capital allocation rule for a class of risk measures including VaR with scenario analysis and the Basel Accord risk measures.
see \citet*{ShiWerker-2012}, \cite*{WPLBS-2013}, \cite*{XCL-2013}, and the references therein for asset allocation methods using VaR and Basel Accord risk measures.}
(6) It is often argued that if a non-subadditive risk measure is used in determining the regulatory capital for a financial institution, then to reduce its regulatory capital, the institution has an incentive to legally break up into various subsidiaries.
However, breaking up an institution into subsidiaries may not be bad, as it prevents the loss of one single business unit from causing the bankruptcy of the entire institution.
On the contrary, if a subadditive risk measure is used, then that institution has an incentive to merge with other financial institutions, which may lead to financial institutions that are too big to fail. Hence, it is not clear by using this type of argument alone whether a risk measure should be subadditive or not.

Even if one believes in subadditivity, VaR (and median shortfall) satisfies subadditivity in most relevant situations. In fact, \citet*{Daielsson-2013} show that VaR (and median shortfall) is subadditive in the relevant tail region if asset returns are regularly varying and possibly dependent, although VaR does not satisfy global subadditivity. \citet{Ibragimov-Walden-2007} and \citet{Ibragimov-2009} show that VaR is subadditive for the infinite variance stable distributions with finite mean. ``In this sense, they showed that VaR is subadditive for the tails of all fat distributions, provided the tails are not super fat (e.g., Cauchy distribution)" (\citet*{Linton-2011}). \citet*{GRT-2007} stress that ``tail thickness required [for VaR] to violate subadditivity, even for small probabilities, remains an extreme situation because it corresponds to such poor conditioning information that expected loss appears to be infinite."

(iv) \cite{EWW2014} argue that ``with respect to dependence uncertainty in aggregation, VaR is less robust compared to expected shortfall" because
VaR is not aggregation-robust but expected shortfall is. However, their counterexample (i.e., their Example 2.1) only shows that $\VaR$ may not be aggregation-robust at the level $\alpha$ such that $F^{-1}(\cdot)$ is not continuous at $\alpha$.
There are only at most a countable number of such $\alpha$; in fact, if $F$ is a continuous distribution, then no such $\alpha$ exists. On the contrary, for any other $\alpha$, VaR at level $\alpha$ is aggregation-robust, because VaR at level $\alpha$ is Hampel-robust and Hampel-robustness implies aggregation-robustness;\footnote{Aggregation-robustness is a notion of robustness that is weaker than Hampel-robustness. By Theorem 2.21 in \citet{Huber-Ronchetti-2009}, a risk measure (statistical functional) $\rho$ is Hampel-robust at a distribution $F$ is essentially equivalent to that $\rho$ is weakly continuous at $F$. More precisely, if $\rho$ is Hampel robust at $F$, then for any $\epsilon>0$, there exists $\delta>0$ such that for $\forall G\in \mathcal{N}_{\delta}(F):=\{H\mid d(F, H)<\delta\}$, it holds that $|\rho(F)-\rho(G)|<\epsilon$. In contrast, $\rho$ is aggregation-robust at $F$ means that for any $\epsilon>0$, there exists $\delta>0$ such that for $\forall G\in \mathcal{N}_{\delta}(F)\cap \mathcal{A}_{F}$, it holds that $|\rho(F)-\rho(G)|<\epsilon$, where $\mathcal{A}_F:=\{H\mid \text{There exist integer}\ m>0\ \text{and random variables}\ X_1, \ldots,\allowbreak X_m,\allowbreak X'_1,\allowbreak \ldots,\allowbreak X'_m,\allowbreak \text{such that}\allowbreak\ X_i\overset{d}{\sim}X'_i, i=1, \ldots, m, \sum_{i=1}^m X_i\overset{d}{\sim}F,\ \text{and}\ \sum_{i=1}^m X'_i\overset{d}{\sim}H.\}$. Since $\mathcal{N}_{\delta}(F)\cap \mathcal{A}_F \varsubsetneq \mathcal{N}_{\delta}(F)$, aggregation robustness is weaker than Hampel-robustness.} note that by Corollary 3.7 of \citet*{CDS-2010} expected shortfall is not Hampel-robust.


(iv) Expected shortfall is more conservative than VaR because $\ES_{\alpha}>\VaR_{\alpha}$. This argument is misleading because $\ES$ at level $\alpha$ should be compared with $\VaR$ at level $(1+\alpha)/2$ (i.e. $\MS$ at level $\alpha$). $\ES_{\alpha}$ may be smaller (i.e., less conservative) than $\MS_{\alpha}$, as mean may be smaller than median. For example, if the tail loss distribution is a Weibull distribution with a shape parameter lager than 3.44, then $\ES_{\alpha}$ is smaller than $\MS_{\alpha}$ (see, e.g., \citet{Hippel-2005}).

\subsection{Other Comments}

It is worth noting that it is not desirable for a risk measure to be too sensitive to the tail risk. For example, let $L$ denote the loss that could occur to a person who walks on the street. There is a very small but positive probability that the person could be hit by a car and lose his life; in that unfortunate case, $L$ may be infinite. Hence, the ES of $L$ may be equal to infinity, suggesting that the person should never walk on the street, which is apparently not reasonable. In contrast, the MS of $L$ is a finite number.


Theorem \ref{thm:elicitable_measures} generalizes the main result in \citet{Ziegel-2013}, which shows the only elicitable spectral risk measure is the mean functional; note that VaR is not a spectral risk measure. \citet{Weber-2006} derives a characterization theorem (Theorem 3.1) for risk measures with convex acceptance set $\mathcal{N}$ and convex rejection set $\mathcal{N}^c$ under two topological conditions on $\mathcal{N}$: (1) there exists $x\in\mathbb{R}$ with $\delta_x\in\mathcal{N}$ such that for $y\in\mathbb{R}$ and $\delta_y\in \mathcal{N}^c$, $(1-\alpha)\delta_x+\alpha \delta_y\in\mathcal{N}$ for sufficiently small $\alpha>0$; (2) $\mathcal{N}$ is $\psi$-weakly closed for some gauge function $\psi: \mathbb{R}\to [1, \infty)$. That characterization theorem cannot be applied in this paper because we do not make any assumption on the forecasting objective function $S(\cdot, \cdot)$ in the definition of elicitability and hence the topological conditions may not hold. For example, the results in \citet{Bellini-Bignozzi-2013}, which rely on the characterization theorem in \citet{Weber-2006},  make strong assumptions on the forecasting objective function $S(\cdot, \cdot)$,\footnote{These assumptions include three conditions in Definition 3.1 and two conditions in Theorem 4.2: (1) $S(x, y)$ is continuous in $y$; (2) for any $x\in[-\epsilon, \epsilon]$ with $\epsilon>0$, $S(x, y)\leq \psi(y)$ for some gauge function $\psi$.} requiring a more restrictive definition of elicitability than \citet{Gneiting-2011}; under their definition, median or quantile may not be elicitable, while they are always elicitable in the sense of \citet{Gneiting-2011}. The elicitability of a risk measure is also related to the statistical theory for the evaluation of probability forecasts (\citet*{LSG-2011}).


The axioms in this paper are based on economic considerations. Other axioms based on mathematical
considerations include convexity (\citet{FollmerSchied2002}, \citet{FG2002, FG-2005}), comonotonic subadditivity (\citet*{SongYan2006, SongYan2009}, \citet*{KPH-2006, KPH-2013}), comonotonic convexity (\citet{SongYan2006, SongYan2009}).
\citet*{DVGKTV-2006} provides a survey on comonotonicity and risk measures.

\appendix



\section{Proof of Lemma \ref{lemma:risk_measure_Schmeidler}}\label{app:proof_rm_rep}
\begin{proof}
%
Without loss of generality, we only need to prove for the case $s=1$, as $\rho$ satisfies Axioms A1-A5 if and only if $\frac{1}{s}\rho$ satisfies Axioms A1-A5 (with $s=1$ in Axiom A3).

The ``only if" part. First, we show that \eqref{equ:e_60} holds for any $X\in\LInf$.
Define the set function $\nu(E):=\rho(1_E), E\in\mathcal{F}$. Then, it follows from Axiom A2 and A3 that $\nu$ is monotonic, $\nu(\emptyset)=0$, and $\nu(\Omega)=1$.
For $M\geq1$, define $\mathcal{L}^M:=\{X\mid |X|\leq M\}$.
For any $X\in\LInf$, let $M_0$ be the essential supremum of $|X|$ and denote
$X^{M_0}:=\min(M_0, \max(X, -M_0))$. Then $X^{M_0}\in\mathcal{L}^{M_0}$ and $X=X^{M_0}$ a.s., which implies that $\rho(X)=\rho(X^{M_0})$ (by Axiom A4) and $\nu(X>x)=\nu(X^{M_0}>x)$, $\forall x$. Since $\rho$ satisfies Axioms A1-A3 on $\LInf$, it follows that $\rho$ satisfies the conditions (i)-(iii)  of
the Corollary in Section 3 of \citet{Schmeidler-1986} (with $B(K)$ in the corollary defined to be $\mathcal{L}^{1+M_0}$). Hence, it follows from the Corollary that
\begin{align}\label{equ:e_61}
\rho(X)&=\rho(X^{M_0})=\int_0^{\infty}\nu( X^{M_0}>x)dx+\int^0_{-\infty}(\nu(X^{M_0}>x)-1)dx\notag\\
&=\int_0^{\infty}\nu(X>x)dx+\int^0_{-\infty}(\nu(X>x)-1)dx.
\end{align}
Let $U$ be a uniform $U(0, 1)$ random variable.
Define the function $h$ such that $h(0)=0$, $h(1)=1$, and $h(p):=\rho(1_{\{U\leq p\}})$, $\forall p\in(0, 1)$. By Axiom A4, $h(\cdot)$ satisfies $\nu(A)=h(P(A))$ for all $A$. Therefore, by \eqref{equ:e_61}, \eqref{equ:e_60} holds for $X$. In addition, for any $0<q<p<1$, $h(p)=\rho(1_{\{U\leq p\}})\geq \rho(1_{\{U\leq q\}})=h(q)$. Hence, $h$ is an increasing function.

Second, we show that \eqref{equ:e_60} holds for any (possibly unbounded) $X\in\X$. For $M>0$, since $X^M$ belongs to $\LInf$, it follows that \eqref{equ:e_60} holds for $X^M$, which implies
\begin{align*}
  \rho(X^M)&=\int_0^{\infty}h(P(X^{M}>x))dx+\int^0_{-\infty}(h(P(X^{M}>x))-1)dx\notag\\
  &=\int_0^{M}h(P(X>x))dx+\int^0_{-M}(h(P(X>x))-1)dx.
\end{align*}
Letting $M\to\infty$ on both sides of the above equation and using Axiom A5, we conclude that \eqref{equ:e_60} holds for $X$.

The ``if" part. Suppose $h$ is a distortion function and $\rho$ is defined by
\eqref{equ:e_60}. Define the set function $\nu(A):=h(P(A)), \forall A\in\mathcal{F}$. Then $\rho(X)$ is the Choquet integral of $X$ with respect to $\nu$. By definition of $\rho$ and simple verification, $\rho$ satisfies Axioms A2-A5.
It follows from \citet[][Proposition 5.1]{Denneberg94} that $\rho$ satisfies positive homogeneity and comonotonic additivity, which implies that $\rho$ satisfies Axiom A1.
%
\end{proof}

\section{Proof of Theorem \ref{thm:elicitable_measures}}\label{app:proof_elict_rm}

First, we give the following definition:\footnote{A similar definition for a set-valued (not single-valued) statistical functional is given in \citet{Osband-1985} and \citet{Gneiting-2011}.
}
\begin{definition}\label{def:convex_ls}
   A single-valued statistical functional $\rho$ is said to have convex level sets with respect to $\PSet$, if for any two distributions $F_1\in\PSet$ and $F_2\in\PSet$, $\rho(F_1)=\rho(F_2)$ implies that  $\rho(\lambda F_1+(1-\lambda)F_2)=\rho(F_1)$, $\forall \lambda\in(0, 1)$.
\end{definition}

The following Lemma \ref{lemma:elicitable_nece} gives a necessary condition for a single-valued statistical functional to be elicitable. The lemma is a variant of Proposition 2.5 of \citet{Osband-1985}, Lemma 1 of \citet*{Lambert-Pennock-Shoham-2008}, and Theorem 6 of \citet{Gneiting-2011}, which concern set-valued statistical functionals.

\begin{lemma}\label{lemma:elicitable_nece}
  If a single-valued statistical functional $\rho$ is elicitable with respect to $\PSet$, then $\rho$ has convex level sets with respect to $\PSet$.
\end{lemma}
\begin{proof} Suppose $\rho$ is elicitable. Then there exists a forecasting objective function $S(x, y)$ such that \eqref{equ:e_62} holds. For any two distribution $F_1$ and $F_2$ and any $\lambda\in(0, 1)$,
denote $F_{\lambda}:=\lambda F_1+(1-\lambda)F_2$. If $t=\rho(F_1)=\rho(F_2)$, then $t=\min\{x\mid x\in \argmin_{x}\int S(x,y)dF_i(y)\}$, $i=1, 2$.
Since $\int S(x,y)dF_{\lambda}(y)=\lambda\int S(x,y)dF_1(y)+(1-\lambda)\int S(x,y)dF_2(y)$, it follows that $t\in \argmin_{x}\int S(x,y)dF_{\lambda}(y)$. For any $t'\in \argmin_{x}\int S(x,y)dF_{\lambda}(y)$, it holds that $\int S(t',y)dF_{\lambda}(y)\leq\int S(t,y)dF_{\lambda}(y)$, which implies that $\lambda\int S(t',y)dF_{1}(y)+(1-\lambda)\int S(t',y)dF_{2}(y)\leq\lambda\int S(t,y)dF_{1}(y)+(1-\lambda)\int S(t,y)dF_{2}(y)$. However, by definition of $t$, $\int S(t,y)dF_{i}(y)\leq \int S(t',y)dF_{i}(y), i=1, 2$. Therefore, $\int S(t,y)dF_{i}(y)=\int S(t',y)dF_{i}(y), i=1, 2$, which implies that
$t'\in\argmin_x \int S(x, y)dF_{i}(y)$, $i=1, 2$. Since $t=\min\{x\mid x\in \argmin_{x}\int S(x,y)dF_i(y)\}$, it follows that $t'\geq t$. Therefore, $t=\min\{x\mid x\in \argmin_{x}\int S(x,y)dF_{\lambda}(y)\}=\rho(F_{\lambda})$.
\end{proof}

\begin{lemma}\label{lemma:inf_c_sup}
Let $c\in[0, 1]$ be a constant. If $\rho$ is defined in \eqref{equ:e_60} with $h(u)=1-c, \forall u\in(0, 1)$, $h(0)=0$, and $h(1)=1$, then $\rho=c\VaR_0+(1-c)\VaR_1$, where $\VaR_0(F):=\inf\{x\mid F(x)>0\}$ and $\VaR_1(F):=\inf\{x\mid F(x)=1\}$. In addition, $\rho$ has convex level sets with respect to $\PSet=\{F\mid  \rho(F)\ \text{is well defined}\}$.
\end{lemma}
\begin{proof} If $\VaR_0(F)\geq 0$, then
\begin{align*}
  \rho(F)&=\int_{(0, \VaR_0(F))}h(1-F(x))\,dx+\int_{(\VaR_0(F), \VaR_1(F))}h(1-F(x))\,dx\\
  &\phantom{=}+\int_{(\VaR_1(F), \infty)}h(1-F(x))\,dx\\
  &=\VaR_0(F)+(1-c)(\VaR_1(F)-\VaR_0(F))=c\VaR_0(F)+(1-c)\VaR_1(F).
\end{align*}
If $\VaR_0(F)<0$, similar calculation also leads to $\rho(F)=c\VaR_0(F)+(1-c)\VaR_1(F)$.

Suppose $t=\rho(F_1)=\rho(F_2)$. Denote $F_{\lambda}:=\lambda F_1+(1-\lambda)F_2$,
$\lambda\in(0, 1)$. There are three cases:

(i) $c=0$. Then, $t=\VaR_1(F_1)=\VaR_1(F_2)$. By definition of $\VaR_1$, $F_i(x)<1$ for $x<t$ and $F_i(x)=1$ for $x\geq t$. Hence,
for any $\lambda\in(0, 1)$, it holds that $F_{\lambda}(x)<1$ for $x<t$ and $F_{\lambda}(x)=1$ for $x\geq t$. Hence, $\rho(F_{\lambda})=\VaR_1(F_{\lambda})=t$.

(ii) $c\in(0, 1)$. Without loss of generality, suppose $\VaR_0(F_1)\leq \VaR_0(F_2)$.
Since $t=c\VaR_0(F_1)+(1-c)\VaR_1(F_1)=c\VaR_0(F_2)+(1-c)\VaR_1(F_2)$, $\VaR_1(F_1)\geq \VaR_1(F_2)$. Hence, for any $\lambda\in(0, 1)$,
$\VaR_0(F_{\lambda})=\VaR_0(F_1)$ and $\VaR_1(F_{\lambda})=\VaR_1(F_1)$. Hence, $\rho(F_{\lambda})=t$.

(iii) $c=1$. Then, $t=\VaR_0(F_1)=\VaR_0(F_2)$. By definition of $\VaR_0$, $F_i(x)=0$ for $x<t$ and $F_i(x)>0$ for $x> t$. Hence,
for any $\lambda\in(0, 1)$, it holds that $F_{\lambda}(x)=0$ for $x<t$ and $F_{\lambda}(x)>0$ for $x> t$. Hence, $\rho(F_{\lambda})=\VaR_0(F_{\lambda})=t$.
\end{proof}

\begin{lemma}\label{lemma:lc_quantile}
  Let $\alpha\in(0, 1)$ and $c\in[0, 1]$.
Let $\rho$ be defined in \eqref{equ:e_60} with $h$ being defined as $h(x):=(1-c)\cdot 1_{\{x=1-\alpha\}}+1_{\{x>1-\alpha\}}$.
  Then
  \begin{equation}\label{equ:e_68}
  \rho(F)=cq_{\alpha}^-(F)+(1-c)q^+_{\alpha}(F),\ \forall F\in\PSet,
  \end{equation}
  where $q_{\alpha}^-(F):=\inf\{x\mid F(x)\geq \alpha\}$ and $q_{\alpha}^+(F):=\inf\{x\mid F(x)> \alpha\}$. Furthermore, $\rho$ has convex level sets with respect to $\PSet=\{F_X\mid X\ \text{is a proper random variable}\}$.
\end{lemma}
\begin{proof} 
%
Define $g(x):=1-h(1-x)$, $x\in[0, 1]$. Then, $g(x)=c\cdot 1_{\{x=\alpha\}}+1_{\{x>\alpha\}}$,
and $\rho$ can be represented as
$$\rho(F)=-\int^0_{-\infty}g(F(x))dx+\int_0^{\infty}(1-g(F(x)))dx.$$
Note that $F(x)=\alpha$ for $x\in[q_{\alpha}^-(F), q_{\alpha}^+(F))$. Consider three cases:

(i) $q_{\alpha}^-(F)\geq 0$. In this case,
\begin{align*}
  \rho(F)&=\int_0^{\infty}(1-g(F(x)))dx\\
  &=\int_{[0, q^-_{\alpha}(F))}(1-g(F(x)))dx+\int_{[q^-_{\alpha}(F), q^+_{\alpha}(F))}(1-g(F(x)))dx+\int_{(q^+_{\alpha}(F), \infty)}(1-g(F(x)))dx\\
  &=q^-_{\alpha}(F)+(1-c)(q^+_{\alpha}(F)-q^-_{\alpha}(F))=cq^-_{\alpha}(F)+(1-c)q^+_{\alpha}(F).
\end{align*}

(ii) $q_{\alpha}^-(F)< 0 < q_{\alpha}^+(F)$. In this case,
\begin{align*}
  \rho(F)&=-\int_{(q_{\alpha}^-(F), 0)}g(F(x))dx+\int_{(0, q^+_{\alpha}(F))}(1-g(F(x)))dx=cq^-_{\alpha}(F)+(1-c)q^+_{\alpha}(F).
\end{align*}

(iii) $q_{\alpha}^+(F)\leq 0$. In this case,
\begin{align*}
  \rho(F)&=-\int_{(-\infty, q^-_{\alpha}(F))}g(F(x))dx-\int_{(q^-_{\alpha}(F), q^+_{\alpha}(F))}g(F(x))dx-\int_{(q^+_{\alpha}(F), 0)}g(F(x))dx\\
  &=-c(q^+_{\alpha}(F)-q^-_{\alpha}(F))+q^+_{\alpha}(F)=cq^-_{\alpha}(F)+(1-c)q^+_{\alpha}(F),
\end{align*}
which completes the proof of \eqref{equ:e_68}.

We then show that $\rho$ has convex level sets with respect to $\PSet$. Suppose that $\rho(F_1)=\rho(F_2)$. Then
\begin{equation}\label{equ:e_a63}
cq_{\alpha}^-(F_1)+(1-c)q^+_{\alpha}(F_1)=cq_{\alpha}^-(F_2)+(1-c)q^+_{\alpha}(F_2).
\end{equation}
For $\lambda\in(0, 1)$, define $F_{\lambda}:=\lambda F_1+(1-\lambda)F_2$. There are three cases:

(i) $c=0$.
Then, $\rho=q_{\alpha}^+$. Denote  $t=q_{\alpha}^+(F_1)=q_{\alpha}^+(F_2)$, then $F_i(x)>\alpha$ for $x>t$ and $F_i(x)\leq\alpha$ for $x<t$, $i=1, 2$. Hence, $F_{\lambda}(x)>\alpha$ for $x>t$ and $F_{\lambda}(x)\leq\alpha$ for $x<t$, which implies $t=q^+_{\alpha}(F_{\lambda})$, i.e., $q^+_{\alpha}$ has convex level sets with respect to $\PSet$.

(ii) $c\in(0, 1)$. Without loss of generality, assume $q_{\alpha}^-(F_1)\geq q_{\alpha}^-(F_2)$. Then it follows from \eqref{equ:e_a63} that $q_{\alpha}^+(F_1)\leq q_{\alpha}^+(F_2)$. Therefore,
$[q_{\alpha}^-(F_1), q_{\alpha}^+(F_1)]\subset [q_{\alpha}^-(F_2), q_{\alpha}^+(F_2)]$. There are two subcases: (ii.i) $q_{\alpha}^-(F_1)<q_{\alpha}^+(F_1)$. In this case, $F_{\lambda}(x)<\alpha$ for $x<q_{\alpha}^-(F_1)$; $F_{\lambda}(x)=\alpha$ for $x\in[q_{\alpha}^-(F_1), q_{\alpha}^+(F_1))$; and
$F_{\lambda}(x)>\alpha$ for $x>q_{\alpha}^+(F_1)$. Therefore, $q_{\alpha}^-(F_{\lambda})=q_{\alpha}^-(F_1)$ and
$q_{\alpha}^+(F_{\lambda})=q_{\alpha}^+(F_1)$, which implies that $\rho(F_{\lambda})=\rho(F_1)$. (ii.ii)  $q_{\alpha}^-(F_1)=q_{\alpha}^+(F_1)$. In this case, $F_{\lambda}(x)<\alpha$ for $x<q_{\alpha}^-(F_1)$ and $F_{\lambda}(x)>\alpha$ for $x>q_{\alpha}^+(F_1)$. Therefore, $q_{\alpha}^-(F_{\lambda})=q_{\alpha}^-(F_1)$ and
$q_{\alpha}^+(F_{\lambda})=q_{\alpha}^+(F_1)$, which implies that $\rho(F_{\lambda})=\rho(F_1)$. Therefore, $\rho$ has convex level sets.

(iii) $c=1$. Then, $\rho=q_{\alpha}^-=\VaR_{\alpha}$. Denote  $t=q_{\alpha}^-(F_1)=q_{\alpha}^-(F_2)$, then
$F_i(x)<\alpha$ for $x<t$ and $F_i(x)\geq \alpha$ for $x\geq t$, $i=1, 2$. Hence, $F_{\lambda}(x)<\alpha$ for $x<t$ and $F_{\lambda}(x)\geq \alpha$ for $x\geq t$, which implies that $q^-_{\alpha}(F_{\lambda})=t$, i.e.,
$q^-_{\alpha}$ has convex level sets with respect to $\PSet$.
\end{proof}

Next, we prove the following Theorem \ref{thm:convex_level_set_rm}, which shows that among the class of risk measures based on Choquet expected utility theory, only four kinds of risk measures satisfy the necessary condition of being elicitable.

\begin{theorem}\label{thm:convex_level_set_rm}
Let $\PSet_0$  be the set of distributions with finite support. Let $h$ be a distortion function defined on $[0, 1]$ and let $\rho(\cdot)$ be defined as in \eqref{equ:e_60}. Then, $\rho(\cdot)$ has convex level sets with respect to $\PSet_0$ if and only if one of the following four cases holds:
  \begin{itemize}
    \item[(i)] There exists $c\in[0, 1]$, such that $\rho=c\VaR_0+(1-c)\VaR_1$, where $\VaR_0(F):=\inf\{x\mid F(x)>0\}$ and $\VaR_1(F):=\inf\{x\mid F(x)=1\}$.

    \item[(ii)]
    There exists $\alpha\in(0, 1)$ such that $\rho(F)=\VaR_{\alpha}(F)$, $\forall F$.

    \item[(iii)]
    There exists $\alpha\in(0, 1)$ and $c\in[0, 1)$ such that
    \begin{equation}\label{equ:e_74}
    \rho(F)=cq_{\alpha}^-(F)+(1-c)q_{\alpha}^+(F),\ \forall F,
    \end{equation}
    where $q_{\alpha}^-(F):=\inf\{x\mid F(x)\geq \alpha\}$ and $q^+_{\alpha}(F):=\inf\{x\mid F(x)>\alpha\}$.
    \item[(iv)]
    $\rho(F)=\int xdF(x)$, $\forall F$.
  \end{itemize}
Furthermore, the risk measures listed above have convex level sets with respect to $\PSet$ defined in Theorem \ref{thm:elicitable_measures}.
\end{theorem}
\begin{proof}[Proof of Theorem \ref{thm:convex_level_set_rm}.]
Define $g(u):=1-h(1-u)$, $u\in[0, 1]$. Then $g(0)=0$, $g(1)=1$, and $g$ is increasing on $[0, 1]$. And then, $\rho$ can be represented as
$$\rho(F)=-\int^0_{-\infty}g(F(x))dx+\int_0^{\infty}(1-g(F(x)))dx.$$
For a discrete distribution $F=\sum_{i=1}^n p_i\delta_{x_i}$, where $0\leq x_1<x_2<\cdots<x_n$, $p_i>0$, $i=1, \ldots, n$, and $\sum_{i=1}^n p_i=1$, it can be shown by simple calculation that $\rho(F)=g(p_1)x_1+\sum_{i=2}^n (g(\sum_{j=1}^i p_j)-g(\sum_{j=1}^{i-1} p_j))x_i$.

There are three cases for $g$:

Case (i): for any $q\in(0, 1)$, $g(q)=0$. Then $g(u)=1_{\{u=1\}}$. By Lemma \ref{lemma:inf_c_sup} (with $c=0$), $\rho=\VaR_1$ and $\rho$ has convex level sets  with respect to $\PSet$.

Case (ii): there exists $q_0\in(0, 1)$ such that $g(q_0)=1$ and $g(q)\in\{0, 1\}$ for all $q\in(0, 1)$.
Let $\alpha=\inf\{q\mid g(q)=1\}$. There are three subcases: (ii.i) $\alpha=0$. Then, $g(u)=1_{\{u>0\}}$. By Lemma \ref{lemma:inf_c_sup} (with $c=1$), $\rho=\VaR_0$ and $\rho$ has convex level sets  with respect to $\PSet$.
(ii.ii) $\alpha\in(0, 1)$ and $g(\alpha)=1$. Then, $g(u)=1_{\{u\geq \alpha\}}$. By Lemma \ref{lemma:lc_quantile} (with $c=1$), $\rho=q^-_{\alpha}=\VaR_{\alpha}$ and $\rho$ has convex level sets with respect to $\PSet$.
(ii.iii) $\alpha\in(0, 1)$ and $g(\alpha)=0$. Then, $g(u)=1_{\{u>\alpha\}}$.
By Lemma \ref{lemma:lc_quantile} (with $c=0$),
$\rho=q_{\alpha}^+$ and $\rho$ has convex level sets with respect to $\PSet$.

Case (iii): there exists $q\in(0, 1)$ such that $g(q)\in (0, 1)$. Suppose $\rho$ has convex level sets with respect to $\PSet_0$. For any
$0< x_1< x_2$ and $q\in(0, 1)$ that satisfy
  \begin{equation}\label{equ:e_a2}
  1=\rho(\delta_1)=\rho(q \delta_{x_1}+(1-q)\delta_{x_2})=x_1 g(q)+x_2 (1-g(q)),
  \end{equation}
since $\rho$ has convex level sets, it follows that
  \begin{equation}\label{equ:e_a3}
  1=\rho(v(q \delta_{x_1}+(1-q)\delta_{x_2})+(1-v)\delta_1),\ \forall v\in(0, 1).
  \end{equation}
For any $q\in(0, 1)$ such that $g(q)\in (0, 1)$, \eqref{equ:e_a2} holds for any $(x_1, x_2)=(1-c, -\frac{g(q)}{1-g(q)}(1-c)+\frac{1}{1-g(q)})$, $\forall c\in(0, 1)$. Noting that $x_1<1<x_2$, \eqref{equ:e_a3} implies
  \begin{align*}
  1={}&\rho(v(q \delta_{x_1}+(1-q)\delta_{x_2})+(1-v)\delta_1)\\
  ={}& x_1 g(vq)+g(vq + 1 -v)-g(vq)+x_2(1-g(vq+1-v))\\
  ={}& (1-c) g(vq)+g(vq + 1 -v)-g(vq)\\
  & +\left[-\frac{g(q)}{1-g(q)}(1-c)+\frac{1}{1-g(q)}\right](1-g(vq+1-v))\\
  ={}& 1+ c\left[-g(vq)+\frac{g(q)}{1-g(q)}(1-g(vq+1-v))\right],\ \forall v\in(0, 1), \forall c\in(0, 1).
  \end{align*}
  Therefore,
  \begin{equation}\label{equ:e_a4}
  -g(vq)+\frac{g(q)}{1-g(q)}(1-g(vq+1-v))=0,\ \forall v\in(0, 1),\forall q\ \text{such that}\ g(q)\in(0, 1).
  \end{equation}
  Let $\alpha=\sup\{q\mid g(q)=0, q\in[0, 1]\}$ and $\beta=\inf\{q\mid g(q)=1, q\in[0, 1]\}$. Since there exists $q_0\in(0,1)$ such that $g(q_0)\in(0, 1)$, it follows that $\alpha\leq q_0<1$, $g(\alpha)\leq g(q_0)<1$, $\beta\geq q_0>0$, and $g(\beta)\geq g(q_0)>0$.

  There are four subcases:

  Case (iii.i) $\alpha=\beta$ and $g(\alpha)=c\in(0, 1)$. In this case, $\alpha=\beta\in(0, 1)$. By the definition of $\alpha$ and $\beta$, $g(x)=0$ for $x<\alpha$ and $g(x)=1$ for $x>\alpha$. By Lemma \ref{lemma:lc_quantile}, $\rho=cq_{\alpha}^-+(1-c)q_{\alpha}^+$ and $\rho$ has convex level sets with respect to $\PSet$.

  Case (iii.ii) $\alpha<\beta$ and $g(\alpha)\in(0, 1)$. In this case, $\alpha\in(0, 1)$. It follows from the definition of $\beta$ that $g((\alpha+\beta)/2)<1$. Let $\epsilon_0=\beta-\alpha$. By the definition of $\beta$, $g(\alpha+\epsilon)<1$ for all $\epsilon\in (0, \epsilon_0)$. In addition, $g(\alpha+\epsilon)\geq g(\alpha)>0$ for all $\epsilon\in (0, \epsilon_0)$. Hence, $g(\alpha+\epsilon)\in(0, 1)$ for all $\epsilon\in (0, \epsilon_0)$.
  For any $\eta\in(0, \alpha)$ and $\epsilon\in(0, \epsilon_0)$, let $q=\alpha + \epsilon$
   and $v=\frac{\alpha-\eta}{\alpha+\epsilon}$. Then it follows from the definition of $\alpha$ that $g(vq)=g(\alpha-\eta)=0$, which implies from \eqref{equ:e_a4} that $1=g(vq+1-v)=g(\alpha-\eta+\frac{\epsilon+\eta}{\alpha+\epsilon})$, for any $\epsilon\in(0, \epsilon_0), \eta\in(0, \alpha)$.
  Then, $g(\alpha+)=\lim_{\epsilon\downarrow 0, \eta\downarrow 0}g(\alpha-\eta+\frac{\epsilon+\eta}{\alpha+\epsilon})=1$, which contradicts to $g(\alpha+)\leq g((\alpha+\beta)/2)<1$. Therefore, this case does not hold.

  Case (iii.iii) $\alpha<\beta$, $g(\alpha)=0$, and $g(\beta)\in(0, 1)$. Since $g(\beta)\in(0, 1)$, it follows that $\beta\in(0, 1)$.
  By the definition of $\beta$,
  for any $\eta\in(0, 1-\beta)$, $g(\beta+\eta)=1$.
  By the definition of $\alpha$, $g((\beta+\alpha)/2)>0$. Hence, $g(\beta-)\geq g((\beta+\alpha)/2)>0$. Hence,
  there exists $\epsilon_0>0$ such that $g(\beta-\epsilon)>0$ for any $\epsilon\in(0, \epsilon_0)$. On the other hand, $g(\beta-\epsilon)\leq g(\beta)<1$ for any $\epsilon\in(0, \epsilon_0)$. Hence, $g(\beta-\epsilon)\in(0, 1)$ for any $\epsilon\in(0, \epsilon_0)$.
  Then, for any $\eta\in(0, 1-\beta)$ and $\epsilon\in(0, \epsilon_0)$, let $q=\beta-\epsilon$ and $v=\frac{1-\beta-\eta}{1-\beta+\epsilon}$. Then, we have $g(vq+1-v)=g(\beta+\eta)=1$.
  Since $g(\beta-\epsilon)\in(0, 1)$ for $\epsilon\in(0, \epsilon_0)$, it follows from \eqref{equ:e_a4} that $0=g(vq)=g(\frac{1-\beta-\eta}{1-\beta+\epsilon}(\beta-\epsilon))$, which implies that $g(\beta-)=\lim_{\eta\downarrow 0,\epsilon\downarrow 0}g(\frac{1-\beta-\eta}{1-\beta+\epsilon}(\beta-\epsilon))=0$. This contradicts to $g(\beta-)>0$. Therefore, this case does not hold.

  Case (iii.iv) $\alpha<\beta$, $g(\alpha)=0$, $g(\beta)=1$. Let $q_0\in(0, 1)$ such that $g(q_0)\in(0, 1)$. Then, $\alpha<q_0<\beta$. We will show that either there exists a constant $c\in(0, 1)$ such that $g(u)=c$, $\forall u\in(0, 1)$, or $g(u)=u$, $\forall u\in(0, 1)$.

  First, we will show that $\alpha=0$ and $\beta=1$. Suppose for the sake of contradiction that $\alpha>0$. Since $\alpha < q_0$, it follows that $g(\alpha+\epsilon)<1$ for all $\epsilon\in(0, \epsilon_0)$, where $\epsilon_0=q_0-\alpha$. Furthermore, by the definition of $\alpha$, $g(\alpha+\epsilon)>0$ for all $\epsilon\in(0, \epsilon_0)$. Hence, $g(\alpha+\epsilon)\in (0, 1)$ for all $\epsilon\in(0, \epsilon_0)$.
  For any $\eta\in(0, \alpha)$ and $\epsilon\in(0, \epsilon_0)$, let $q=\alpha + \epsilon$
   and $v=\frac{\alpha-\eta}{\alpha+\epsilon}$. Then it follows from the definition of $\alpha$ that $g(vq)=g(\alpha-\eta)=0$, which implies from \eqref{equ:e_a4} that $1=g(vq+1-v)=g(\alpha-\eta+\frac{\epsilon+\eta}{\alpha+\epsilon})$, for any $\epsilon\in(0, \epsilon_0), \eta\in(0, \alpha)$.
  Then, $g(\alpha+)=\lim_{\epsilon\downarrow 0, \eta\downarrow 0}g(\alpha-\eta+\frac{\epsilon+\eta}{\alpha+\epsilon})=1$, which contradicts to $g(\alpha+)\leq g(q_0)<1$.
  Therefore, $\alpha=0$.

  In addition, suppose for the sake of contradiction that $\beta<1$. Then, by the definition of $\beta$,
  for any $\eta\in(0, 1-\beta)$, $g(\beta+\eta)=1$.
  Let $\epsilon_0=\beta-q_0$. Since $\beta>q_0$, $g(\beta-\epsilon)\geq g(q_0)>0$ for any $\epsilon\in(0, \epsilon_0)$. By the definition of $\beta$, $g(\beta-\epsilon)<1$ for any $\epsilon\in(0, \epsilon_0)$. Hence, $g(\beta-\epsilon)\in(0, 1)$ for any $\epsilon\in(0, \epsilon_0)$.
  Then, for any $\eta\in(0, 1-\beta)$ and $\epsilon\in(0, \epsilon_0)$, let $q=\beta-\epsilon$ and $v=\frac{1-\beta-\eta}{1-\beta+\epsilon}$. Then, we have $g(vq+1-v)=g(\beta+\eta)=1$.
  Since $g(\beta-\epsilon)\in(0, 1)$ for any $\epsilon\in(0, \epsilon_0)$, it follows from \eqref{equ:e_a4} that $0=g(vq)=g(\frac{1-\beta-\eta}{1-\beta+\epsilon}(\beta-\epsilon))$, which implies that $g(\beta-)=\lim_{\eta\downarrow 0,\epsilon\downarrow 0}g(\frac{1-\beta-\eta}{1-\beta+\epsilon}(\beta-\epsilon))=0$. This contradicts to that $g(\beta-)\geq g(q_0)>0$. Therefore, $\beta=1$.

  Then, it follows from $\alpha=0$ and $
  \beta=1$ that
  \begin{equation}\label{equ:e_a5}
  g(q)\in(0, 1),\ \forall q\in(0, 1).
  \end{equation}
  Therefore, it follows from \eqref{equ:e_a4} and \eqref{equ:e_a5} that
  \begin{equation}\label{equ:e_a6}
  -g(vq)+\frac{g(q)}{1-g(q)}(1-g(vq+1-v))=0,\ \forall v\in(0, 1),\forall q\in(0, 1).
  \end{equation}
For any $q\in(0, 1)$ and $v\in(0, 1)$, $vq+1-v>q$ and $\lim_{v\uparrow 1}(vq+1-v)=q$. It then follows from \eqref{equ:e_a6} that
  \begin{align}
    g(q-)&=\lim_{v\uparrow 1}g(vq)=\lim_{v\uparrow 1}\frac{g(q)}{1-g(q)}(1-g(vq+1-v))= \frac{g(q)}{1-g(q)}(1-g(q+)),\ \forall q\in (0, 1).\label{equ:e_63}
  \end{align}

Second, we consider two cases for $g$:

Case (iii.iv.i) There exist $0<u_1<u_2<1$ such that $g(u_1)=g(u_2)$. Let $w_1=\inf\{u\mid g(u)=g(u_1)\}$ and $w_2=\sup\{u\mid g(u)=g(u_2)\}$. Consider three further cases: (a) $w_1>0$. Since $\lim_{q\downarrow w_1}\frac{1-u_2}{1-q}=\frac{1-u_2}{1-w_1}<1=\lim_{q\downarrow w_1}\frac{w_1}{q}$,
there exists $q_0\in(w_1, u_2)$ such that
$\frac{1-u_2}{1-q_0}<\frac{w_1}{q_0}$. Choose $v_0\in(0, 1)$ such that $\frac{1-u_2}{1-q_0}<v_0<\frac{w_1}{q_0}$. Since $v_0q_0<w_1$, $g(v_0q_0)<g(u_1)$. And, since $w_1<q_0<v_0q_0+1-v_0<u_2$, $g(q_0)=g(v_0q_0+1-v_0)=g(u_1)$. Therefore,
$
  -g(v_0q_0)+\frac{g(q_0)}{1-g(q_0)}(1-g(v_0q_0+1-v_0))>0
$,
which contradicts to \eqref{equ:e_a6}. Hence, this case cannot hold. (b) $w_2<1$. Since $\lim_{q\uparrow w_2}\frac{1-w_2}{1-q}=1>\frac{u_1}{w_2}=\lim_{q\uparrow w_2}\frac{u_1}{q}$,
there exists $q_0\in(u_1, w_2)$ such that
$\frac{1-w_2}{1-q_0}>\frac{u_1}{q_0}$. Choose $v_0\in(0, 1)$ such that $\frac{1-w_2}{1-q_0}>v_0>\frac{u_1}{q_0}$. Since $w_2>q_0>v_0q_0>u_1$, $g(q_0)=g(v_0q_0)=g(u_1)$. And, since $v_0q_0+1-v_0>w_2$, $g(v_0q_0+1-v_0)>g(u_1)$. Therefore,
$
  -g(v_0q_0)+\frac{g(q_0)}{1-g(q_0)}(1-g(v_0q_0+1-v_0))<0
$,
which contradicts to \eqref{equ:e_a6}. Hence, this case cannot hold. (c) $w_1=0$ and $w_2=1$. In this case, $g(u)=c, \forall u\in(0, 1)$, for some constant $c\in(0, 1)$. By Lemma \ref{lemma:inf_c_sup}, $\rho=c\VaR_0+(1-c)\VaR_1$, and $\rho$ has convex level sets with respect to $\PSet$.

Case (iii.iv.ii) $g$ is strictly increasing on $(0, 1)$. Then, $g(p_1)-g(p_2)\neq 0$ for any $p_1\neq p_2$. We will show that $g(1-)=1$ and $g(0+)=0$. Consider $0<x_1<x_2<x_3$ and $p_1, p_2\in (0, 1)$ such that
\begin{equation*}
\rho(p_1\delta_{x_1}+(1-p_1)\delta_{x_2})=\rho(p_2\delta_{x_1}+(1-p_2)\delta_{x_3}),
\end{equation*}
which is equivalent to
\begin{equation}\label{equ:e_a11}
x_1g(p_1)+x_2(1-g(p_1))=x_1 g(p_2)+(1-g(p_2))x_3.
\end{equation}
Let $\frac{x_1}{x_2}=c_1$ and $\frac{x_3}{x_2}=c_3$. Then, $c_1\in(0, 1)$, $c_3>1$, and \eqref{equ:e_a11} is equivalent to
\begin{equation}\label{equ:e_a12}
c_1=\frac{1-g(p_2)}{g(p_1)-g(p_2)}c_3-\frac{1-g(p_1)}{g(p_1)-g(p_2)}.
\end{equation}
For any fixed $0<p_1<p_2<1$ and $1<c_3<\frac{1-g(p_1)}{1-g(p_2)}$, define $c_1$ as in \eqref{equ:e_a12}. Then, $c_1\in(0, 1)$. For any such $p_1, p_2, c_3$, and $c_1$, it follows from the convexity of the level sets of $\rho$ that
\begin{align*}
& x_1g(p_1)+x_2(1-g(p_1))=\rho(p_1\delta_{x_1}+(1-p_1)\delta_{x_2})\\
={}& \rho(v(p_1\delta_{x_1}+(1-p_1)\delta_{x_2})+(1-v)(p_2\delta_{x_1}+(1-p_2)\delta_{x_3}))\\
={}&\rho((vp_1+(1-v)p_2)\delta_{x_1}+v(1-p_1)\delta_{x_2}+(1-v)(1-p_2)\delta_{x_3})\\
={}&x_1g(vp_1+(1-v)p_2)+x_2(g(v+(1-v)p_2)-g(vp_1+(1-v)p_2))\\
&+x_3(1-g(v+(1-v)p_2)),\ \forall  v\in(0, 1),
\end{align*}
which is equivalent to
\begin{align*}
& c_1[g(p_1)-g(vp_1+(1-v)p_2)]+1-g(p_1)-g(v+(1-v)p_2)+g(vp_1+(1-v)p_2)\\
={}& c_3[1-g(v+(1-v)p_2)],\ \forall  v\in(0, 1).
\end{align*}
Plugging \eqref{equ:e_a12} into the above equation, we obtain that for any $0<p_1<p_2<1$, any $1<c_3<\frac{1-g(p_1)}{1-g(p_2)}$, and any $v\in(0, 1)$, it holds that
\begin{align}\label{equ:e_a13}
0={}&  c_3\left[\frac{1-g(p_2)}{g(p_1)-g(p_2)}(g(p_1)-g(vp_1+(1-v)p_2))-1+g(v+(1-v)p_2)\right]\notag\\
& -\frac{1-g(p_1)}{g(p_1)-g(p_2)}[g(p_1)-g(vp_1+(1-v)p_2)]+1-g(p_1)\notag\\
&-g(v+(1-v)p_2)+g(vp_1+(1-v)p_2).
\end{align}
Therefore,
\begin{align*}
0=&-\frac{1-g(p_1)}{g(p_1)-g(p_2)}[g(p_1)-g(vp_1+(1-v)p_2)]+1-g(p_1)\\
&-g(v+(1-v)p_2)+g(vp_1+(1-v)p_2),\ \forall v\in(0, 1), \forall p_1<p_2,
\end{align*}
which is equivalent to
\begin{align}\label{equ:e_a14}
0={}& g(vp_1+(1-v)p_2)(1-g(p_2))+g(v+(1-v)p_2)(g(p_2)-g(p_1))\notag\\
&+g(p_1)g(p_2)-g(p_2),\ \forall v\in(0, 1), \forall p_1<p_2.
\end{align}
Letting $v\uparrow 1$ in \eqref{equ:e_a14}, we obtain
\begin{equation}\label{equ:e_64}
0=g(p_1+)(1-g(p_2))+g(1-)(g(p_2)-g(p_1))+g(p_1)g(p_2)-g(p_2), \forall p_1<p_2.
\end{equation}
Since $g$ is increasing on $(0, 1)$, there exists $p^*_1\in(0, 1)$, such that $g$ is continuous at $p^*_1$. Choose any $p^*_2>p_1^*$. Letting $p_1=p_1^*$ and $p_2=p_2^*$ in \eqref{equ:e_64} leads to
$(g(p^*_1)-g(p^*_2))(1-g(1-))=0$. Since $g$ is strictly increasing, it follows that
\begin{equation}\label{equ:e_a15}
g(1-)=1.
\end{equation}
Letting $q=\frac{1}{2}$ in \eqref{equ:e_a6} leads to
\begin{equation}\label{equ:e_a16}
\frac{g(\frac{v}{2})}{1-g(1-\frac{v}{2})}=\frac{g(\frac{1}{2})}{1-g(\frac{1}{2})}, \forall v\in(0, 1).
\end{equation}
It follows from \eqref{equ:e_a16} and \eqref{equ:e_a15} that
\begin{align}\label{equ:e_a61}
g(0+)&=\lim_{v\downarrow 0}g(\frac{v}{2})=\lim_{v\downarrow 0}\frac{g(\frac{1}{2})}{1-g(\frac{1}{2})}(1-g(1-\frac{v}{2}))=\frac{g(\frac{1}{2})}{1-g(\frac{1}{2})}(1-g(1-))=0.
\end{align}
We will then show that $g$ is continuous on $(0, 1)$. By \eqref{equ:e_a6}, we have
\begin{align}\label{equ:e_a30}
  &g(v-)=\lim_{q\uparrow 1}g(vq)=\lim_{q\uparrow 1}\frac{g(q)}{1-g(q)}(1-g(vq+1-v))\notag\\
  ={}& \lim_{q\uparrow 1}g(q)\lim_{q\uparrow 1}\frac{1-g(vq+1-v)}{1-g(q)}\notag\\
={}& g(1-)\lim_{q\uparrow 1}\frac{1-g(vq+1-v)}{g((1-q)v)}\frac{g((1-q)v)}{g(1-q)}\frac{g(1-q)}{1-g(q)}\notag\\
={}& \lim_{q\uparrow 1}
\frac{1-g(\frac{1}{2})}{g(\frac{1}{2})}\frac{g((1-q)v)}{g(1-q)}\frac{g(\frac{1}{2})}{1-g(\frac{1}{2})}\ (\text{by}\ \eqref{equ:e_a15}\ \text{and}\ \eqref{equ:e_a16})\notag\\
={}&
\lim_{q\uparrow 1}\frac{g((1-q)v)}{g(1-q)}=\lim_{q\downarrow 0}\frac{g(qv)}{g(q)},\ \forall v\in(0, 1).
\end{align}

Now consider $0=x_1<x_2<x_3<x_4$ and $p_1, p_2\in(0, 1)$ such that
\begin{equation*}
\rho(p_1\delta_{x_1}+(1-p_1)\delta_{x_3})=\rho(p_2\delta_{x_2}+(1-p_2)\delta_{x_4}),
\end{equation*}
which is equivalent to
\begin{equation}\label{equ:e_a17}
x_1g(p_1)+x_3(1-g(p_1))=x_2g(p_2)+x_4(1-g(p_2)).
\end{equation}
Since $\rho$ has convex level sets, it follows that for any $v\in(0, 1)$, it holds that
\begin{align}\label{equ:e_a19}
 & x_3(1-g(p_1))=x_1g(p_1)+x_3(1-g(p_1))=\rho(p_1\delta_{x_1}+(1-p_1)\delta_{x_3})\notag\\
 ={}&{} \rho(v(p_1\delta_{x_1}+(1-p_1)\delta_{x_3})+(1-v)(p_2\delta_{x_2}+(1-p_2)\delta_{x_4}))\notag\\
 ={}&{}
 \rho(vp_1\delta_{x_1}+(1-v)p_2\delta_{x_2}+v(1-p_1)\delta_{x_3}+(1-v)(1-p_2)\delta_{x_4})\notag\\ ={}&{}
 x_2(g(vp_1+(1-v)p_2)-g(vp_1))\notag\\
 &{}+x_3(g(v+(1-v)p_2)-g(vp_1+(1-v)p_2))+x_4(1-g(v+(1-v)p_2)).
\end{align}
Let $\frac{x_3}{x_2}=1+c_3$ and $\frac{x_4}{x_2}=1+c_3+c_4$. Then, $c_3>0$, $c_4>0$, and \eqref{equ:e_a17} becomes
\begin{equation}\label{equ:e_a18}
c_3=\frac{1-g(p_2)}{g(p_2)-g(p_1)}c_4+\frac{g(p_1)}{g(p_2)-g(p_1)}.
\end{equation}
Furthermore, \eqref{equ:e_a19} is equivalent to
\begin{align}\label{equ:e_a20}
0={}& g(vp_1+(1-v)p_2)-g(vp_1)+(1+c_3+c_4)(1-g(v+(1-v)p_2))\notag\\
&+(1+c_3)(g(v+(1-v)p_2)-g(vp_1+(1-v)p_2)-1+g(p_1)), \forall v\in(0, 1).
\end{align}
For any $0<p_1<p_2<1$ and $c_4>0$, let $c_3$ be defined in \eqref{equ:e_a18}. Then, $c_3>0$. Hence, \eqref{equ:e_a20} holds for any such $p_1, p_2, c_3$, and $c_4$. Plugging \eqref{equ:e_a18} into \eqref{equ:e_a20}, we obtain that for any $0<p_1<p_2<1$ and any $c_4>0$, it holds that
\begin{align}\label{equ:e_a21}
0={}& g(vp_1+(1-v)p_2)-g(vp_1)+\frac{g(p_2)}{g(p_2)-g(p_1)}[g(p_1)-g(vp_1+(1-v)p_2)]\notag\\
& + c_4\frac{1-g(p_2)}{g(p_2)-g(p_1)}[g(v+(1-v)p_2)-g(vp_1+(1-v)p_2)-1+g(p_1)]\notag\\
& +c_4 \frac{1-g(p_1)}{g(p_2)-g(p_1)}[1-g(v+(1-v)p_2)], \forall v\in(0, 1),
\end{align}
which implies that
\begin{align*}
0={}&g(vp_1+(1-v)p_2)-g(vp_1)\\
&+\frac{g(p_2)}{g(p_2)-g(p_1)}[g(p_1)-g(vp_1+(1-v)p_2)],\ \forall 0<p_1<p_2<1, \forall v\in(0, 1),
\end{align*}
which can be simplified to be
\begin{equation}\label{equ:e_a43}
-g(vp_1+(1-v)p_2)-(g(p_2)-g(p_1))\frac{g(vp_1)}{g(p_1)}+g(p_2)=0,\ \forall p_1<p_2, \forall v\in(0, 1).
\end{equation}
Letting $p_2\uparrow 1$ in \eqref{equ:e_a43} and applying \eqref{equ:e_a15}, we obtain
\begin{equation}\label{equ:e_65}
-g((vp_1+1-v)-)-(1-g(p_1))\frac{g(vp_1)}{g(p_1)}+1=0,\ \forall 0<p_1<1, \forall v\in(0, 1).
\end{equation}
Then, it follows from \eqref{equ:e_a6} and \eqref{equ:e_65} that
$$g((vp_1+1-v)-)=g(vp_1+1-v), \forall 0<p_1<1, \forall v\in(0, 1),$$
which implies that
\begin{equation}\label{equ:e_66}
g(v-)=g(v), \forall v\in(0, 1).
\end{equation}
It follows from \eqref{equ:e_63} and \eqref{equ:e_66} that $g$ is continuous on $(0, 1)$, i.e.,
\begin{equation}\label{equ:e_67}
g(v-)=g(v)=g(v+), \forall v\in(0, 1).
\end{equation}

Lastly, we will show that $g(u)=u$ for any $u\in (0, 1)$. Letting $p_1\downarrow 0$  in \eqref{equ:e_a43}, we obtain
\begin{equation}\label{equ:e_a62}
-g(((1-v)p_2)+)-(g(p_2)-g(0+))\lim_{p_1\downarrow 0}\frac{g(vp_1)}{g(p_1)}+g(p_2)=0,\ \forall 0<p_2<1, \forall v\in(0, 1).
\end{equation}
Applying \eqref{equ:e_a61}, \eqref{equ:e_a30}, and \eqref{equ:e_67} to \eqref{equ:e_a62}, we obtain
\begin{equation}\label{equ:e_a23}
g((1-v)p_2)=g(p_2)(1-g(v)),\ \forall 0<p_2<1, \forall v\in (0, 1).
\end{equation}
Letting $p_2\uparrow 1$ in \eqref{equ:e_a23} and using \eqref{equ:e_a15} and \eqref{equ:e_67}, we obtain
\begin{equation}\label{equ:e_a24}
g(1-v)=g(1-)(1-g(v))=1-g(v),\ \forall v\in(0, 1),
\end{equation}
which in combination with \eqref{equ:e_a23} implies
\begin{equation}\label{equ:e_a25}
g(vp_2)=g(v)g(p_2),\ \forall 0<p_2<1, \forall v\in(0, 1).
\end{equation}
In the following, we will show by induction that
\begin{equation}\label{equ:e_a32}
g(\frac{k}{2^n})=\frac{k}{2^n},\ k=1,2,\ldots, 2^n-1, \forall n\in \mathbb{N}.
\end{equation}
Letting $v=\frac{1}{2}$ in \eqref{equ:e_a24}, we obtain $g(\frac{1}{2})=\frac{1}{2}$. Hence, \eqref{equ:e_a32} holds for $n=1$. Suppose \eqref{equ:e_a32} holds for $n$. We will show that it also holds for $n+1$. In fact, for any $0\leq k\leq 2^{n-1}-1$, since $1\leq 2k+1\leq 2^n-1$, it follows from \eqref{equ:e_a25} that
\begin{equation}\label{equ:e_a33}
g(\frac{2k+1}{2^{n+1}})=g(\frac{1}{2})g(\frac{2k+1}{2^{n}})=\frac{2k+1}{2^{n+1}},\ 0\leq k\leq 2^{n-1}-1.
\end{equation}
For any $2^{n-1}\leq k\leq 2^n-1$, it holds that $1\leq 2^{n+1}-(2k+1)\leq 2^n-1$. 
Hence, it follows from \eqref{equ:e_a24} that
\begin{align}\label{equ:e_a34}
 & g(\frac{2k+1}{2^{n+1}})=1-g(\frac{2^{n+1}-(2k+1)}{2^{n+1}})= 1- \frac{2^{n+1}-(2k+1)}{2^{n+1}}\ (\text{by}\ \eqref{equ:e_a33})\notag\\
 = & \frac{2k+1}{2^{n+1}},\ 2^{n-1}\leq k\leq 2^n-1.
\end{align}
In addition, for any $1\leq k\leq 2^n-1$, $g(\frac{2k}{2^{n+1}})=g(\frac{k}{2^{n}})=\frac{k}{2^{n}}$, which in combination with \eqref{equ:e_a33} and \eqref{equ:e_a34} implies that \eqref{equ:e_a32} holds for $n+1$, and hence holds for any $n$. Since $\{k/2^n,k=1,\ldots, 2^n-1, n\in\mathbb{N}\}$ is dense on $(0, 1)$ and $g$ is continuous on $(0, 1)$, it follows from \eqref{equ:e_a32} that $g(u)=u$ for all $u\in(0, 1)$, which completes the proof.
\end{proof}

Finally, the proof of Theorem \ref{thm:elicitable_measures} is as follows.

\begin{proof}[Proof of Theorem \ref{thm:elicitable_measures}.] By Lemma \ref{lemma:elicitable_nece} and Theorem \ref{thm:convex_level_set_rm}, only those risk measures listed in cases (i)-(iv) of Theorem \ref{thm:convex_level_set_rm} satisfy the necessary condition for being an elicitable risk measure. Therefore, we only need to study the elicitability of those risk measures.

First, we will show that for $c\in(0, 1]$, $\rho=c\VaR_0+(1-c)\VaR_1$ is not elicitable. Suppose for the sake of contradiction that $\rho$ is elicitable, then there exists a function $S$ such that \eqref{equ:e_62} holds. For any $u$, letting $F=\delta_{u}$ in \eqref{equ:e_62} and noting $\rho(\delta_u)=u$ yields
\begin{equation}\label{equ:e_95}
S(u, u)\leq S(x, u), \forall x, \forall u,\ \text{and the equality holds only if}\ u\leq x.
\end{equation}
For any $u<v$ and $p\in(0, 1)$, letting $F=p\delta_u+(1-p)\delta_{v}$ in \eqref{equ:e_62} yields $pS(cu+(1-c)v, u)+(1-p)S(cu+(1-c)v, v)\leq pS(x, u)+(1-p)S(x, v)$, $\forall x$. Letting $p\to 0$ leads to
\begin{equation}\label{equ:e_96}
S(cu+(1-c)v, v)\leq S(x, v), \forall u < v, \forall x.
\end{equation}
Letting $x=v$ in \eqref{equ:e_96}, we obtain
\begin{equation}\label{equ:e_97}
S(cu+(1-c)v, v)\leq S(v, v), \forall u < v.
\end{equation}
By \eqref{equ:e_95}, $S(v, v)\leq S(cu+(1-c)v, v)$, $\forall u<v$, which in combination with \eqref{equ:e_97} implies
$S(v, v)=S(cu+(1-c)v, v)$, $\forall u<v$; however, by \eqref{equ:e_95}, $S(v, v)=S(cu+(1-c)v, v)$ implies $v\leq cu+(1-c)v$, which contradicts to $u<v$. Hence, $\rho$ is not elicitable.


Second, we will show that for $c=0$, $\rho=c\VaR_0+(1-c)\VaR_1=\VaR_1$ is elicitable with respect to $\PSet$. Let $a>0$ be a constant and define the forecasting objective function
$$S(x, y)=\begin{cases}
  0, & \text{if}\ x\geq y,\\
  a, & \text{else}.
\end{cases}$$
Then for any $F\in\PSet$ and any $x\geq \rho(F)$,
$$\int_{\mathbb{R}} S(x, y)dF(y)=\int_{y\leq \rho(F)} S(x, y)dF(y)=0.$$
On the other hand, for any $F\in\PSet$ and any $x<\rho(F)$,
$$\int_{\mathbb{R}} S(x, y)dF(y)=\int_{x<y\leq \rho(F)} S(x, y)dF(y)=a\int_{x<y\leq \rho(F)}dF(y)=a(1-F(x))>0.$$
Therefore, for any $F\in\PSet$, $\rho(F)=\min\{x\mid x\in \argmin_x \int S(x, y)dF(y)\}$.

Third, we will show that for any $\alpha\in(0, 1)$, $\VaR_{\alpha}$ is elicitable with respect to $\PSet$. Let $g(\cdot)$ be a strictly increasing function defined on $\mathbb{R}$. Define
\begin{equation}\label{equ:e_51}
S(x, y)=(1_{\{x\geq y\}}-\alpha)(g(x)-g(y)).
\end{equation}
and define $\PSet=\{F_X\mid  E|g(X)|<\infty\}$.\footnote{For example, if $g(x):=x$, then $\PSet=\{F_X\mid X\in\mathcal{L}^1(\Omega, \mathcal{F}, P)\}$; if $g(x):=x^{\frac{1}{2n+1}}$ ($n\geq 1$), then $\PSet$ includes heavy tailed distributions with infinite mean such as Cauchy distribution.}
It follows from Theorem 9 in \cite{Gneiting-2011} that 
$$[q^-_{\alpha}(F), q^+_{\alpha}(F)]=\argmin_{x}\int S(x, y)dF(y),$$
where $q^-_{\alpha}(F):=\inf\{y\mid F(y)\geq \alpha\}$ and $q^+_{\alpha}(F):=\inf\{y\mid F(y)>\alpha\}$. Therefore, $\VaR_{\alpha}(F)=q^-_{\alpha}(F)$ satisfies \eqref{equ:e_62} with $S$ defined in \eqref{equ:e_51}.

Fourth, we will show that $\rho$ defined in \eqref{equ:e_74} is not elicitable with respect to $\PSet$. Suppose for the purpose of contradiction that $\rho$ is elicitable. Fix any $a>0$ and denote $I:=(-a, a)$. Let $\PSet_I$ be the set of probability measures that have strictly positive probability density on the interval $I$ and whose support is $I$. Then since $\PSet_I\subset \PSet$ and $\rho$ is elicitable with respect to $\PSet$, $\rho$ is also elicitable with respect to $\PSet_I$. Therefore, there exists a forecasting objective function $S(x, y)$ such that
$$\rho(F)=\min\{x\mid x\in \argmin_{x}\int S(x, y)dF(y)\},\forall F\in \PSet_I.$$
For any $F\in\PSet_I$, the equation $F(x)=\alpha$ has a unique solution $q_{\alpha}(F)$ and $q_{\alpha}^-(F)=q_{\alpha}(F)=q_{\alpha}^+(F)$. Hence, $\rho(F)=q_{\alpha}(F), \forall F\in\PSet_I$. Therefore, we have
$$q_{\alpha}(F)\in \argmin_x \int S(x, y)dF(y), \forall F\in\PSet_I.$$
Then, it follows from the proposition in \citet[][p. 372]{Thomson-1979} that\footnote{\citet{Thomson-1979} obtains the proposition for the case when the interval $I=(-\infty, \infty)$; in our case, $I=(-a, a)$. It can be verified that the proof of the proposition in \citet{Thomson-1979} can be easily adapted to the case of $I=(-a, a)$. The details are available from the authors upon request. To make the paper self-contained, the Proposition of \citet[][p. 372]{Thomson-1979} is quoted here: A necessary and sufficient condition for a scheme $H$ to elicit $x^*$, solution of ``$F(x)=r$", as best answer is that $H$ satisfies:
\begin{equation}
H(x, y)=\begin{cases}
  A_1(x)+B_1(y)\ a.e.\ \text{if}\ y\leq x\\
  A_2(x)+B_2(y)\ a.e.\ \text{if}\ y> x
\end{cases}
\end{equation}
with
\begin{equation}\label{equ:e_75}
(A_1(x_1)-A_1(x_2))r + (A_2(x_1)-A_2(x_2))(1-r)=0,\ \forall x_1, x_2,
\end{equation}
and
\begin{align}
  &B_2(\cdot)-B_1(\cdot)\ \text{is non-increasing a.e.}\\
  &B_2(\cdot)-B_1(\cdot)+A_2(\cdot)-A_1(\cdot)=0\ \text{a.e.}
\end{align}} there exist measurable functions $A_1$, $A_2$, $B_1$, and $B_2$ such that
\begin{equation}\label{equ:e_69}
S(x, y)=\begin{cases}
  A_1(x)+B_1(y)\ a.e.\ \text{if}\ y\leq x,\\
  A_2(x)+B_2(y)\ a.e.\ \text{if}\ y> x,
\end{cases}
\end{equation}
and
\begin{equation}\label{equ:e_70}
(A_1(x_1)-A_1(x_2))\alpha + (A_2(x_1)-A_2(x_2))(1-\alpha)=0,\ \forall x_1, x_2\in I.
\end{equation}
Choose a distribution $F_0\in\PSet$ such that
$q^-_{\alpha}(F_0)<q^+_{\alpha}(F_0)$, $F_0$ has a density $f_0$ that satisfies $f_0(x)=0$ for $x\in(q^-_{\alpha}(F_0), q^+_{\alpha}(F_0))$, and $F_0(q^-_{\alpha}(F_0))=F_0(q^+_{\alpha}(F_0))=\alpha$. Then, it follows from \eqref{equ:e_69} that for any $x\in[q^-_{\alpha}(F_0), q^+_{\alpha}(F_0)]$,
\begin{align}\label{equ:e_71}
  &\int S(x, y)dF_0(y)=\int_{y\leq x} S(x, y)f_0(y)dy+\int_{y>x} S(x, y)f_0(y)dy\notag\\
  ={}&\int_{y\leq x} (A_1(x)+B_1(y))f_0(y)dy+\int_{y>x} (A_2(x)+B_2(y))f_0(y)dy\notag\\
  ={}&A_1(x)\int_{y\leq x}f_0(y)dy+\int_{y\leq x}B_1(y)f_0(y)dy+A_2(x)\int_{y>x} f_0(y)dy+\int_{y>x}B_2(y) f_0(y)dy\notag\\
  ={}&A_1(x)\alpha+\int_{y\leq x}B_1(y)f_0(y)dy+A_2(x)(1-\alpha)+\int_{y>x}B_2(y) f_0(y)dy. 
\end{align}
Since $f_0(x)=0$ for $x\in(q^-_{\alpha}(F_0), q^+_{\alpha}(F_0))$, it follows that
\begin{align}
  \int_{y\leq x_1}B_1(y)f_0(y)dy&=\int_{y\leq x_2}B_1(y)f_0(y)dy,\ \forall x_1, x_2\in [q^-_{\alpha}(F_0), q^+_{\alpha}(F_0)],\label{equ:e_72}\\
  \int_{y> x_1}B_2(y)f_0(y)dy&=\int_{y> x_2}B_2(y)f_0(y)dy,\ \forall x_1, x_2\in [q^-_{\alpha}(F_0), q^+_{\alpha}(F_0)].\label{equ:e_73}
\end{align}
Since $c\in[0, 1)$, $\rho(F_0)=cq^-_{\alpha}(F_0)+(1-c)q^+_{\alpha}(F_0)\in(q^-_{\alpha}(F_0), q^+_{\alpha}(F_0)]$. It then follows from \eqref{equ:e_70}, \eqref{equ:e_71}, \eqref{equ:e_72}, and \eqref{equ:e_73} that
for any $x\in[q^-_{\alpha}(F_0), q^+_{\alpha}(F_0)]$,
\begin{align*}
&\int S(x, y)dF_0(y)-\int S(\rho(F_0), y)dF_0(y)\\
={}&(A_1(x)-A_1(\rho(F_0)))\alpha+(A_2(x)-A_2(\rho(F_0)))(1-\alpha)=0,
\end{align*}
which in combination with
$\rho(F_0)\in\argmin_x \int S(x, y)dF_0(y)$
implies that
$$[q^-_{\alpha}(F_0), q^+_{\alpha}(F_0)]\subset \argmin_{x}\int S(x, y)dF_0(y).$$
Therefore,
$$\rho(F_0)=\min\{x\mid x\in \argmin_x \int S(x, y)dF_0(y)\}\leq q^-_{\alpha}(F_0),$$
which contradicts to
$\rho(F_0)>q^-_{\alpha}(F_0)$. Hence,
$\rho$ defined in \eqref{equ:e_74} is not elicitable.

Fifth, it follows from Theorem 7 in \citet{Gneiting-2011} that $\rho(F):=\int xdF(x)$ is elicitable with respect to $\PSet$. The proof is thus completed.
\end{proof}

\bibliographystyle{dcu}
\bibliography{RiskMeasures}

\end{document}